\documentclass[epjST]{svjour}
\usepackage{graphics}
\usepackage{hyperref}
\usepackage{bbm}
\usepackage[utf8]{inputenc}
\usepackage{array}

\usepackage{amsmath}

\usepackage{amsfonts}
\usepackage{amssymb}
\usepackage{subeqnarray}
\usepackage{tensor}
\usepackage{indentfirst} 
\newcommand{\lambdabp}{\stackrel{\neg}{\lambda}_{h}(p^{\mu})}
\newcommand{\lambdap}{\lambda_{h}(p^{\mu})}

\newcommand{\gm}{\gamma_{\mu}}
\begin{document}
\title{Remarks on mass dimension one fermions: The underlying aspects, bilinear forms, Spinor Classification and RIM decomposition}
\author{R. J. Bueno Rogerio\inst{1}\fnmsep\thanks{\email{rodolforogerio@unifei.edu.br}} \and C. H. Coronado Villalobos\inst{2} \and D. Beghetto\inst{3} \and A. R. Aguirre\inst{1} }
\institute{Universidade Federal de Itajub\'a, Instituto de F\'isica e Qu\'imica - IFQ/UNIFEI,  \\
Av. BPS 1303, CEP 37500-903, Itajub\'a, MG, Brazil. \and Instituto Nacional de Pesquisas Espaciais (INPE) Ci\^encias Espaciais e Atmosf\'ericas\\
CEP 12227-010, S\~ao Jos\'e dos Campos, SP, Brazil. \and Universidade Estadual Paulista ``Julio de Mesquita Filho'' UNESP - Campus de Guaratinguetá - DFQ, Avenida Dr. Ariberto Pereira da Cunha, 333, CEP 12516-410, Guaratinguetá, SP, Brazil.}
\abstract{In the present essay we review the underlying physical information behind the first concrete example describing a mass dimension one fermion - namely Elko spinors. We start the program exploring the physical information by evaluating the Elko bilinear forms, both within the proper orthochronous Lorentz subgroup as well as within the VSR theory. As we shall see, such structures do not hold the right observance of the Fierz-Pauli-Kofink quadratic relations. Thus, by the aforementioned reasons, we develop a deformation of the Clifford algebra basis. Such protocol can be accomplished by taking precisely the right Elko’s dual structure during the construction of the bilinear forms related to these spinors. With the appropriated bilinear forms at hands, we search for a real physical interpretation in order to achieve a deeper understanding of such spinor fields. Aiming an interesting application, we present a relation concerning Elko spinors and the neutrino physics via the Heisenberg non-linear theory by means of a bijective linear map between Elko spinors and the so-called Restricted Inomata-McKinley (RIM) spinors. Thus, we describe some of its properties. Some interesting results concerning the construction of RIM-decomposable spinors emerge from such prescription. 
} 
\maketitle
\section{Opening: Looking Beyond the Standard Model}
The so-called Elko spinors compose a new set of spinors with a complex and interesting structure on its own. Historically, they were proposed by Ahluwalia and Grumiller when studying properties of Majorana spinors. Although Elko spinors share with Majorana spinors the property of being eingenspinors of the charge conjugation operator, $C$, they have dual helicity and can take positive (self-conjugated) and negative (anti-self-conjugated) eigenvalues of $C$, whilst the Majorana ones take only the positive value.

From the physical point of view, Elko spinors are constructed to be invisible to the other fields (i.e., it does not couple with the fields of the Standard Model, except for the Higgs boson), becoming a natural candidate to dark matter \cite{jcap}. Mathematically, the dual helicity peculiarity forces one to redefine the adjoint structure, as can be seem in \cite{jcap}. This idiosyncrasy reflects itself when constructing the spin sums for the Elko spinors, which breaks Lorentz symmetry. However, such a pathology can be easily circumvented if one decides to work on the Very Special Relativity framework \cite{cohen} - the spin sums are invariant under transformations of VSR groups \cite{cohen,horvath,ahluwaliahorvath}. There are several areas in which Elko spinors have been studied, from accelerator physics \cite{fen1,fen2,fen3,fen4,rogerio2019,rogerioalternative} to cosmology \cite{co1,co2,co3,co4,co5,co6,co7,co8,co9,co10,co11} and mathematical physics \cite{dualtipo4,beyondlounesto,restrictedinomata,juliodual}. In particular, the appreciation of new dual structures bring interesting possibilities within the algebraic scope \cite{cjr,rcd}. 

As it is well known, much of the physics associated to spinor fields is unveiled from its bilinear covariants by the simple reason that single fermions are not directly experienced. In this context it is indeed important to pay special attention to the subtleties of Clifford algebra when associating real numbers to the bilinear covariants \cite{lounestolivro}. It may sound as a secondary issue, but in fact the opposite is true. In two outstanding papers in the nineties \cite{crawford1,crawford2}, Crawford worked out several important formalizations concerning the bispinor algebra. Among these, a rigorous procedure made to obtain real bilinear covariants was developed. One of the aims of this paper is to make use of this procedure to envisage what (if any) bilinear covariants are real when dealing with mass dimension one spinors. To the best of our hope, the results to be shown here may shed some light on the observables associated to these spinors. With suitable, but important, changes we take advantage of the formalism developed in \cite{crawford1} in order to study the bilinear covariants associated to the Elko spinor case. After a complete analysis, including the right observance of the Fierz-Pauli-Kofink (FPK) \cite{FPK,beres} relations, we arrive at the subset of real bilinear covariants, which, in the lights of \cite{beyondlounesto}, can be fully interpreted. 

Quite recently, new possibilities concerning a redefinition on the adjoint field was deeply investigated \cite{aaca,vicinity}. These formalizations may lead to a local and full Lorentz spin $1/2$ field also endowed with mass dimension one, clearly evading the Weinberg's no-go theorem \cite{nogo}. We have investigated the bilinear forms on this case, however it leads to similar conclusions, i.e., the physical information encoded on the spinors remain unchanged. 
Remarkably enough, concerning to the physical aspects of its recent formulation and its insertion in the irreducible representation of the Poincar\'e symmetries we call attention to the Wigner's  work, where the physical content is supported by the Hilbert space under the Poincar\'e group action \cite{Wigner1,Wigner2} was found consistently one particle states \cite{controversy}. The triumph of the aforementioned Wigner's works is the emergence of hidden particle classes --- it turns out that the particle studied in \cite{aaca} behaves, under discrete symmetries, in a way predicted in one of these hidden cases.  

As an application, we take advantage of some peculiarities of Elko spinors and decompose (map) it into RIM (Restricted Inomata-McKinley) spinors. The so called Inomata-McKinley spinors are a particular class of solutions of the non-linear Heisenberg equation \cite{akira}. A subclass of Inomata-McKinley spinors called Restricted Inomata-McKinley (RIM) spinors was revealed to be useful in describing neutrino physics \cite{novello}. It is well known that free linear massive (or massless) Dirac fields can be represented as a combination of RIM-spinors \cite{novello}. Moreover, it was shown \cite{RIM} that such Dirac spinors are necessarily type-1 in the so-called Lounesto classification, and that they are all non exotic spinors, i.e., the spacetime itself needs to have an underlying trivial topology\footnote{By trivial topology we mean that the fundamental group $\pi_1(M)$ associated to the spacetime manifold $M$ is trivial, i.e., $\pi_1(M) = 0$. A non trivial topology, thus, means that $\pi_1(M) \neq 0$.} in order to enable the very existence of RIM-spinors. Thus, the decomposition in terms of RIM-spinors itself is not allowed in a spacetime with non trivial topology.

The very idea of mapping Elko and Dirac spinor fields is not new \cite{map1,hoffdirac,Julio}. However, the works developed towards this proposal use Elko as being a type-5 spinor field within Lounesto classification, taking the bilinear covariants associated to this class as fundamental elements in the construction of the map. It is well known that Elko fields do not fulfill the requisites to fit in the Lounesto Classification (for more details, please, check for \cite{bilineares}), since their dual structure is defined in a different way than the usually imposed in such classification. Moreover, Elko spinors are governed by a whole non-usual dynamics, carrying a new and different physical content. Then, a true map between Dirac and Elko spinors is, in fact, a map between different spinor spaces. Thus, to transcend the need of using the bilinear structures associated to the spinors would be welcome in such attempt to construct the aforementioned map. Here we will present a construction of such a map in the case of those Dirac fields representable in terms of RIM spinors \cite{RIM}.

This paper is organized as follows: In Section \ref{entrada} we introduce the formal aspects of the Elko spinors. Thus, Sect.\ref{secaobilineares} is reserved for a short and elementary review concerning the main concepts on the bilinear forms, Lounesto's classification and Fierz-Pauli-Kofink (FPK) identities. In Section \ref{elkoforms} we perform a deep analysis on Elko's bilinear forms. Both these sections are related to Elko spinors as objects belonging to orthochronous proper Lorentz subgroup. In Sect.\ref{secaofpk} we use the very ideia of Clifford algebra basis deformation aiming to provide to the Elko spinors the right observance of the FPK identities via the redefinition of the bilinear forms. Regarding to the physical observables, we shall finalize our analysis by defining how the Elko's bilinear forms transform under Lorentz transformations. Finally, in Sect.\ref{secaoHeisenberg} we provide some details and applications, bijectively mapping Elko spinors into RIM spinors. Such a procedure provides interesting outcomes as the possibility to write Elkos spinors in terms of Dirac or RIM spinors and \emph{vice-versa}. We leave to the Appendix \ref{apendiceA} an analogous investigation as developed in Sect. \ref{secaobilineares}, however, now taking into account Elko spinors within VSR framework. Appendix \ref{apendiceB} is devoted to some FPK identities transformations.

\section{Foreword: Main concepts on the Elko spinors formal structure}\label{entrada}

This section is reserved for bookkeeping purposes. Here we shall describe the basic introductory elements concerning spinors and its transformations laws - remarking the main aspects that will be carried along the paper. 

\subsection{Setting up the notation and an overview on spinors}\label{conceitos}
To obtain an explicit form of a given $\psi(p^{\mu})$ spinor we first call our attention to the rest spinors, $\psi(k^{\mu})$. For an arbitrary momentum $(p^{\mu})$, we have the following condition
\begin{equation}\label{1}
\psi(p^{\mu}) = e^{i\kappa.\varphi}\psi(k^{\mu}),
\end{equation}
where the $\psi(k^{\mu})$ rest frame spinor is a direct sum of the $(1/2,0)$ and $(0,1/2)$ Weyl spinor, which usually is defined as 
\begin{equation} \label{espinorqualquer}
\psi(k^{\mu}) = \left(\begin{array}{c}
\phi_R(k^{\mu}) \\ 
\phi_L(k^{\mu})
\end{array} \right).
\end{equation}
Note that we define the $k^{\mu}$ rest frame momentum as 
\begin{equation}
k^{\mu}\stackrel{def}{=}\bigg(m,\; \lim_{p\rightarrow 0}\frac{\boldsymbol{p}}{p}\bigg), \; p=|\boldsymbol{p}|,
\end{equation}
moreover, the general four-momentum (in spherical coordinates) is
\begin{equation}
p^{\mu}=(E, p\sin\theta\cos\phi, p\sin\theta\sin\phi, p\cos\theta).
\end{equation}
Thus, the boost operator is defined as follows
\begin{eqnarray}\label{boostoperator}
e^{i\kappa.\varphi} = \sqrt{\frac{E+m}{2m}}\left(\begin{array}{cc}
\mathbbm{1}+\frac{\vec{\sigma}.\hat{p}}{E+m} & 0 \\ 
0 & \mathbbm{1}-\frac{\vec{\sigma}.\hat{p}}{E+m}
\end{array} \right),
\end{eqnarray}
this yields $\cosh\varphi = E/m, \sinh\varphi=p/m$ with $\hat{\boldsymbol{\varphi}} = \hat{\boldsymbol{p}}$.

Thus, such momentum parametrization allow us to defined the right-hand and left-hand components, in the rest-frame referential, under inspection of the helicity operator it directly provide
\begin{equation}\label{helicidadee}
\vec{\sigma}\cdot\hat{p}\; \phi^{\pm}(k^{\mu}) = \pm \phi^{\pm}(k^{\mu}).
\end{equation}
Thus, the positive helicity component is given by 
\begin{equation}
\phi^{+}(k^{\mu}) = \sqrt{m}e^{ i\vartheta_{1}}\left(\begin{array}{c}
\cos(\theta/2)e^{-i\phi/2} \\ 
\sin(\theta/2)e^{i\phi/2}
\end{array}\right), 
\end{equation} 
and the negative helicity reads
\begin{equation}
\phi^{-}(k^{\mu}) = \sqrt{m}e^{ i\vartheta_{2}}\left(\begin{array}{c}
\sin(\theta/2)e^{-i\phi/2} \\ 
-\cos(\theta/2)e^{i\phi/2}
\end{array}\right).
\end{equation} 

As remarked in Ref \cite{mdobook}, the presence of the phase factor becomes necessary to set up the framework of eigenpinors of parity or charge conjugation operators or eigenspinors of parity operator. One can verify that under a rotation by an angle $\vartheta$ the Dirac spinors pick up a global phase $e^{\pm i\vartheta/2}$, depending on the related helicity. However, this only happens for eigenspinors of parity operator. For the eigenspinors of charge conjugation operator, the phases factor must be $\vartheta_1=0$ and $\vartheta_2=\pi$ \cite{aaca}. Thus, this judicious phase combination ensure for the field the character of locality - further details can be found in \cite{mdobook}. Note that we do not make any allusive comment concerning how is the link among the spinorial component in the introduced spinor in \eqref{espinorqualquer}. As a matter of fact, it is well-known that if parity is introduced as a fundamental link between the $(0,j)$ and $(j,0)$ representation spaces, automatically the Dirac dynamic is reached \cite{ryder,speranca}. 

\subsection{Introducing Elko spinors}
Now, guided by \cite{aaca,ramond}, let us introduce an element to link the $(0,1/2)$ and $(1/2,0)$ representation spaces as follows 
\begin{eqnarray}
\lambda(k^{\mu}) = \left(\begin{array}{c}
\zeta\Theta\phi^*_L(k^{\mu}) \\ 
\phi_L(k^{\mu})
\end{array} \right)
\end{eqnarray} 
where $\zeta$ stands for a constant phase factor, to be further determined, while $\Theta$ is the Wigner time-reversal operator, which in the spin $1/2$ representation read \cite{jcap,ramond}
\begin{eqnarray}
\Theta = \left(\begin{array}{cc}
0 & -1 \\ 
1 & 0
\end{array}\right),
\end{eqnarray} 
such operator carry the following interesting characteristic
\begin{equation}\label{mmp}
\Theta\sigma\Theta^{-1} = -\sigma^*,
\end{equation}  
where $\sigma$ stands for the Pauli matrices. Considering the transformation properties of the right-hand and left-hand  spinors, we can see the following: if $\phi_R(p^\mu)$ transforms as $(1/2,0)$, then $[\zeta\Theta \,\phi_R^\ast(p^\mu)]$ transforms as $(0,1/2)$ spinor, and similarly if $\phi_L(p^\mu)$ transforms as $(0,1/2)$, then $[\zeta\Theta\, \phi_L^\ast(p^\mu)]$ transforms as $(1/2,0)$ spinor, as observed in \cite{ramond}.

Looking towards defining the precise value for the phase $\zeta$, one would impose to the spinors to hold conjugacy under charge conjugation operator. Such task is accomplished by 
\begin{eqnarray}\label{conjugacaoelko}
C\lambda^{S/A}(p^{\mu}) = \pm\lambda^{S/A}(p^{\mu}),
\end{eqnarray}
where $C$ is written in the fashion
\begin{eqnarray}\label{cc}
C = \left(\begin{array}{cc}
\mathbb{O} & i\Theta \\ 
-i\Theta & \mathbb{O}
\end{array} \right)\mathcal{K},
\end{eqnarray}
in which $\mathcal{K}$ is the algebraic complex conjugation. The relation \eqref{conjugacaoelko} is immediately satisfied if the following restriction $\zeta=\pm i$ is taken into account. Notice that Elko spinors form a complete set of eigenspinors of $C$ with positive ($S$) and negative ($A$) eigenvalues under action of $C$. 

A parenthetic remark related with Elko neutrality under action of charge conjugation operator is highlighted here. Without introducing the associated quantum field --- which can be seen in Ref \cite{mdobook} --- here we make a pause and weave allusive comments concerning its possible interactions. The interactions of the such fermions are restricted to dimension-four quartic self-interaction and also to a dimension-four coupling with Higgs boson. This observations opens up the possibility that the new field provides a natural self-interacting dark matter candidate. Thus, the \emph{darkness} arises from two related facts: first, due to mass dimension mismatch of the Standard Model fermions and the mass dimension one fermions the latter cannot enter the SM doublets; and second, the formalism for mass dimension one fermions does not support the SM local gauge interactions \cite{aaca}.

 Thus, immediately follows that the four rest spinors are shaped by two \emph{self-conjugate}
\begin{equation}\label{Elkos}
\lambda^S_{\lbrace-,+\rbrace}(k^{\mu}) = \left(\begin{array}{c}
+i\Theta[\phi^+_L(k^{\mu})]^* \\ 
\phi^+_L(k^{\mu})
\end{array}\right) , \qquad  \lambda^S_{\lbrace+,-\rbrace}(k^{\mu}) = \left(\begin{array}{c}
+i\Theta[\phi^-_L(k^{\mu})]^* \\ 
\phi^-_L(k^{\mu})
\end{array}\right), 
\end{equation} 
and the other two are \emph{anti-self-conjugate}
\begin{equation}\label{Elkoa}
\lambda^A_{\lbrace-,+\rbrace}(k^{\mu}) = \left(\begin{array}{c}
-i\Theta[\phi^+_L(k^{\mu})]^* \\ 
\phi^+_L(k^{\mu})
\end{array}\right) , \qquad  \lambda^A_{\lbrace+,-\rbrace}(k^{\mu}) = \left(\begin{array}{c}
-i\Theta[\phi^-_L(k^{\mu})]^* \\ 
\phi^-_L(k^{\mu})
\end{array}\right).  
\end{equation}
We remark that the lower index stands for the spinor's helicity, the first entry is related with the right-hand component helicity and the second one stands for the left-hand component helicity. An important remark related to such spinors is that they are constructed in such a way that carry dual helicity. To explicit show the aforementioned feature, we first complex conjugate Eq. \eqref{helicidadee} 
\begin{equation}\label{equ:3.13}
\sigma^*\cdot\hat{p} [\phi_L^{\pm}(k^{\mu})]^* = \pm [\phi_L^{\pm}(k^{\mu})]^* .
\end{equation}
Replacing $\sigma^*$ by the relation in Eq. \eqref{mmp}
\begin{equation}
\Theta\sigma\Theta^{-1} \hat{p}[\phi_L^{\pm}(k^{\mu})]^* = \mp [\phi_L^{\pm}(k^{\mu})]^*,
\end{equation}
now, taking into account $\Theta^{-1} = -\Theta$, one find
\begin{equation}
-\Theta\sigma\Theta \hat{p}[\phi_L^{\pm}(k^{\mu})]^* = \mp [\phi_L^{\pm}(k^{\mu})]^*,
\end{equation}
such relation can be displayed as
\begin{equation}\label{helicdireita}
\Theta^{-1}\sigma\Theta \hat{p}[\phi_L^{\pm}(k^{\mu})]^* = \mp [\phi_L^{\pm}(k^{\mu})]^*.
\end{equation}
Finally, we multiply by the left-hand side both sides of \eqref{helicdireita} by $\Theta$, then, we conclude
\begin{equation}\label{equ:3.17}
\sigma\cdot\hat{p}\Theta[\phi_L^{\pm}(k^{\mu})]^* = \mp \Theta[\phi_L^{\pm}(k^{\mu})]^*.
\end{equation}
A quick inspection of $\Theta[\phi_L^{\pm}(k^{\mu})]^*$ says that it carry opposite helicity  when compared with $\phi_L^{\pm}(k^{\mu})$, such outcome strongly contrast with Dirac spinors which are endowed by single-helicity feature \cite{jcap}.

To define Elko spinors in a arbitrary momentum referencial, we may act with the boost operator \eqref{boostoperator} over the spinors, in the rest-frame, i.e.,
\begin{eqnarray}
\lambda^{S/A}_{\lbrace\pm,\mp\rbrace}(p^{\mu})= e^{i\kappa\varphi}\lambda^{S/A}_{\lbrace\pm,\mp\rbrace}(k^{\mu}). 
\end{eqnarray}
Thus, for the \emph{self-conjugated} fields it provides
\begin{equation}\label{equ:3.24}
\lambda^S_{\{-,+\}}(p^{\mu}) = \sqrt\frac{E+m}{2m}\left(1 - \frac{p}{E+m}\right)\lambda^S_{\{-,+\}}(k^{\mu}),
\end{equation}
and 
\begin{equation}\label{equ:3.25}
\lambda^S_{\{+,-\}}(p^{\mu}) = \sqrt\frac{E+m}{2m}\left(1 + \frac{p}{E+m}\right)\lambda^S_{\{+,-\}}(k^{\mu}).
\end{equation}
while for the \emph{anti-self-conjugated} ones we have
 \begin{equation}\label{equ:3.26}
\lambda^A_{\{-,+\}}(p^{\mu}) = \sqrt\frac{E+m}{2m}\left(1 - \frac{p}{E+m}\right)\lambda^A_{\{-,+\}}(k^{\mu}),
\end{equation}
and
\begin{equation}\label{equ:3.27}
\lambda^A_{\{+,-\}}(p^{\mu}) = \sqrt\frac{E+m}{2m}\left(1 + \frac{p}{E+m}\right)\lambda^A_{\{+,-\}}(k^{\mu}).
\end{equation} 
Henceforth, to summarize the notation, we choose to define the boost factor by
\begin{eqnarray}
\mathcal{B}_{\pm} = \sqrt{\frac{E+m}{2}}\bigg(1 \pm \frac{p}{E+m}\bigg).
\end{eqnarray}

Albeit the introduced objects hold a fermionic (and spinorial) character, they do not satisfy the Dirac equation
\begin{eqnarray}
\label{diracELKO21}&&\gamma_{\mu}p^{\mu} \lambda^S_{\{-,+\}}(p^{\mu}) = -im \lambda^S_{\{+,-\}}(p^{\mu}),\\
\label{diracELKO22}&&\gamma_{\mu}p^{\mu} \lambda^S_{\{+,-\}}(p^{\mu}) = im \lambda^S_{\{-,+\}}(p^{\mu}),\\                                                                                      
\label{diracELKO23}&&\gamma_{\mu}p^{\mu} \lambda^A_{\{+,-\}}(p^{\mu}) = -im \lambda^A_{\{-,+\}}(p^{\mu}),\\
\label{diracELKO24}&&\gamma_{\mu}p^{\mu} \lambda^A_{\{-,+\}}(p^{\mu}) = im \lambda^A_{\{+,-\}}(p^{\mu}).                                                                                                                                                                                                                                                                             \end{eqnarray}
the set of relation above can be summarized as
\begin{eqnarray}\label{tipodirac1}
(\gamma_{\mu}p^{\mu} \pm \Xi(p^{\mu})m)\lambda^{S/A}_{\alpha}(p^{\mu})=0,
\end{eqnarray}
note that $[\gamma_{\mu}p^{\mu}, \Xi(p^{\mu})]=0$, thus, it is also possible to write 
\begin{eqnarray}\label{tipodirac2}
(\gamma_{\mu}p^{\mu}\Xi(p^{\mu}) \pm m)\lambda^{S/A}_{\alpha}(p^{\mu})=0.
\end{eqnarray}
However, with the above relations at hand, it is possible to show that Elko fulfil the Klein-Gordon equation
\begin{equation}\label{kgelko}
(\Box + m^2) \lambda^{S/A}_{\alpha}(p^{\mu}) = 0,
\end{equation}
such a feature is a hint towards the Elko mass dimensionality.  

If one wishes to evaluate the Elko's norm under the Dirac's dual structure the following results may be faced
\begin{eqnarray}
\label{sub1}\bar{\lambda}^S_{\{\pm,\mp\}}(p^{\mu}) \lambda^S_{\{\pm,\mp\}}(p^{\mu}) = 0,
\\
\label{sub2}\bar{\lambda}^S_{\{\pm,\mp\}}(p^{\mu}) \lambda^A_{\{\pm,\mp\}}(p^{\mu}) = 0,
\\
\label{sub3}\bar{\lambda}^S_{\{\pm,\mp\}}(p^{\mu}) \lambda^A_{\{\pm,\mp\}}(p^{\mu}) = 0,
\\
\label{sub4}\bar{\lambda}^A_{\{\pm,\mp\}}(p^{\mu}) \lambda^A_{\{\pm,\mp\}}(p^{\mu}) = 0,
\\
\label{sub5}\bar{\lambda}^A_{\{\pm,\mp\}}(p^{\mu}) \lambda^S_{\{\pm,\mp\}}(p^{\mu}) = 0, 
\\
\label{sub6}\bar{\lambda}^A_{\{\pm,\mp\}}(p^{\mu}) \lambda^S_{\{\pm,\mp\}}(p^{\mu}) = 0, 
\end{eqnarray}
besides the imaginary ones \cite{aaca}
\begin{eqnarray}
\label{sub7}\bar{\lambda}^S_{(\pm,\mp)}(p^{\mu}) \lambda^S_{(\mp,\pm)}(p^{\mu}) = \mp 2im,
\\
\label{sub8}\bar{\lambda}^A_{(\pm,\mp)}(p^{\mu}) \lambda^A_{(\mp,\pm)}(p^{\mu}) = \pm 2im.
\end{eqnarray}
Such results above force us to abandon the Dirac structure and deal with a more properly definition of the adjoint. 
We can define the dual by a general formula
\begin{eqnarray}
\stackrel{\neg}{\lambda}\;^{S/A}_{h}(p^{\mu}) =[\Xi(p^{\mu})\lambda^{S/A}_{h}(p^{\mu})]^{\dag}\gamma_0,
\end{eqnarray}
and the relations \eqref{sub1}-\eqref{sub8} explicitly indicates a definition of the operator $\Xi(p^{\mu})$ by
\begin{eqnarray}\label{XI}
\Xi(p^{\mu}) \equiv \frac{1}{2m}&\Big(& \lambda^{S}_{\lbrace +-\rbrace}(p^{\mu})\bar{\lambda}^{S}_{\lbrace +-\rbrace}(p^{\mu}) + \lambda^{S}_{\lbrace -+\rbrace}(p^{\mu})\bar{\lambda}^{S}_{\lbrace -+\rbrace}(p^{\mu}) 
\nonumber\\ 
&-& \lambda^{A}_{\lbrace +-\rbrace}(p^{\mu})\bar{\lambda}^{A}_{\lbrace +-\rbrace}(p^{\mu})-\lambda^{A}_{\lbrace -+\rbrace}(p^{\mu})\bar{\lambda}^{A}_{\lbrace -+\rbrace}(p^{\mu}) \Big),
\end{eqnarray} 
which in its matricial form reads
\begin{eqnarray}\label{xioperator}
 \Xi(p^{\mu})= \left(\begin {array}{cccc} {\frac {ip\sin \left( \theta
 \right) }{m}}&{\frac {-i \left( E+p\cos \left( \theta \right) 
 \right) {{\rm e}^{-i\phi}}}{m}}&0&0\\ \noalign{\medskip}{\frac {i
 \left( E-p\cos \left( \theta \right)  \right) {{\rm e}^{i\phi}}}{m}}&
{\frac {-ip\sin \left( \theta \right) }{m}}&0&0\\ \noalign{\medskip}0&0
&{\frac {-ip\sin \left( \theta \right) }{m}}&{\frac {-i \left( E-p\cos
 \left( \theta \right)  \right) {{\rm e}^{-i\phi}}}{m}}
\\ \noalign{\medskip}0&0&{\frac {i \left( E+p\cos \left( \theta
 \right)  \right) {{\rm e}^{i\phi}}}{m}}&{\frac {ip\sin \left( \theta
 \right) }{m}}\end {array} \right).\nonumber\\ 
\end{eqnarray}
Note that $\Xi^{2}(p^{\mu}) = \mathbbm{1}$, and $\Xi^{-1}(p^{\mu})$ indeed exists and is equal to $\Xi(p^{\mu})$ itself \cite{aaca}. The action of the $\Xi(p^{\mu})$ over the Elko spinors, provides the following set of relations
\begin{eqnarray}
&&\Xi(p^{\mu})\lambda^S_{\{-,+\}}(p^{\mu})=+i\lambda^S_{\{+,-\}}(p^{\mu}), \\
&&\Xi(p^{\mu})\lambda^S_{\{+,-\}}(p^{\mu})=-i\lambda^S_{\{-,+\}}(p^{\mu}), \\
&&\Xi(p^{\mu})\lambda^A_{\{-,+\}}(p^{\mu})=+i\lambda^A_{\{+,-\}}(p^{\mu}),\\
&&\Xi(p^{\mu})\lambda^A_{\{+,-\}}(p^{\mu})=-i\lambda^A_{\{-,+\}}(p^{\mu}).
\end{eqnarray}
Now the explicit form for the dual can be readily written as 
\begin{eqnarray}
\stackrel{\neg}{\lambda}\;^{S/A}_{\lbrace-,+\rbrace}(p^{\mu}) = +i[\lambda^{S/A}_{\lbrace+,-\rbrace}(p^{\mu})]^{\dag}\gamma_0,
 \\\nonumber\\
\stackrel{\neg}{\lambda}\;^{S/A}_{\lbrace+,-\rbrace}(p^{\mu}) = -i[\lambda^{S/A}_{\lbrace-,+\rbrace}(p^{\mu})]^{\dag}\gamma_0.
\end{eqnarray} 
The above defined Elko's dual structure hold the orthonormal relations 
\begin{eqnarray}
&&\label{orto1}\stackrel{\neg}{\lambda}^{S}_{\alpha}(p^{\mu})\lambda^{S}_{\alpha^\prime}(p^{\mu}) = 2m \delta_{\alpha\alpha^{\prime}},
\\
&&\label{orto2}\stackrel{\neg}{\lambda}^{A}_{\alpha}(p^{\mu})\lambda^{A}_{\alpha^\prime}(p^{\mu}) = -2m \delta_{\alpha\alpha^{\prime}},
\\
&&\stackrel{\neg}{\lambda}^{S}_{\alpha}(p^{\mu})\lambda^{A}_{\alpha^\prime}(p^{\mu}) = \stackrel{\neg}{\lambda}^{A}_{\alpha}(p^{\mu})\lambda^{S}_{\alpha^\prime}(p^{\mu})=0. 
\end{eqnarray}

Moreover, the aforementioned results leads, after a bit of calculation, to the following spin sums
\begin{eqnarray}
\sum_h \lambda^{S/A}_h(p^\mu)\stackrel{\neg}{\lambda}\;^{S/A}_h(p^\mu)=\pm m[\mathbbm{1}\pm \mathcal{G}(\phi)],\label{4}
\end{eqnarray} 
where 
\begin{eqnarray}
\mathcal{G}(\phi) = \left(\begin{array}{cccc}
0 & 0 & 0 & -ie^{-i\phi} \\ 
0 & 0 & ie^{i\phi} & 0 \\ 
0 & -ie^{-i\phi} & 0 & 0 \\ 
ie^{i\phi} & 0 & 0 & 0
\end{array}  \right).
\end{eqnarray}
It is important to remark that 
\begin{eqnarray}
\mathcal{G}(\phi)\lambda^{S}_{\alpha}(p^{\mu}) = +\lambda^{S}_{\alpha}(p^{\mu}), \quad \mathcal{G}(\phi)\lambda^{A}_{\alpha}(p^{\mu}) = -\lambda^{A}_{\alpha}(p^{\mu}),
\end{eqnarray} 
as an interesting relation related with such operator is that it may be written as $\mathcal{G}(\phi) = m^{-1}\gamma_{\mu}p^{\mu}\Xi(p^{\mu})$. Nonetheless, note that the Elko spinors are eigenspinors of the $\mathcal{G}(\phi)$ operator, however, such relations above can not be faced as a kinematic relation due the fact that it does not carry time-dependence.   
As highlighted above, the spin sums carry a term which is $\phi$-dependent, revealing a subtle break in the Lorentz covariance, which may forbid and make it difficult to achieve some physical interpretations. Here is where the usual Elko construction stops.  

\subsection{A guide for Elko's adjoint redefinition}
An attempt to solve the Elko breaking Lorentz covariance that emerges in the spin sums \cite{vicinity}, is given by a redefinition in the dual structure, as it can be seen in Ref. \cite{aaca}. Such redefinition reads
\begin{eqnarray}
\stackrel{\neg}{\lambda}^{S}_h(p^\mu) &\rightarrow& \stackrel{\neg}{\lambda}\;^{S}_h(p^\mu)\mathcal{A}, \\
\stackrel{\neg}{\lambda}^{A}_h(p^\mu) &\rightarrow &\stackrel{\neg}{\lambda}\;^{A}_h(p^\mu)\mathcal{B},
\end{eqnarray}
where the operators $\mathcal{A}$ and $\mathcal{B}$ demand some important properties: the spinors $\lambda^{S}_h(p^\mu)$ and $\lambda^{A}_h(p^\mu)$ must be eigenspinors of $\mathcal{A}$ and $\mathcal{B}$ respectively, with eigenvalues given by the unity:
\begin{eqnarray}\label{7}
\mathcal{A}\lambda^{S}_h(p^\mu) = \lambda^{S}_h(p^\mu), \quad\quad \mathcal{B}\lambda^{A}_h(p^\mu)=\lambda^{A}_h(p^\mu).
\end{eqnarray} 
Besides, such operators must fulfill 
\begin{eqnarray}\label{8}
\stackrel{\neg}{\lambda}^{S}_h(p^\mu)\mathcal{A}\lambda^{A}_h(p^\mu) = 0, \quad\quad \stackrel{\neg}{\lambda}^{A}_h(p^\mu)\mathcal{B}\lambda^{S}_h(p^\mu) = 0.
\end{eqnarray}
The set of equations \eqref{7} and \eqref{8} ensures the accuracy of the orthonormality relations, as remarked in \cite{aaca}, to remain unchanged. With the new dual structure in mind, one can evaluate the spin sums, which now will take the operators $\mathcal{A}$ and $\mathcal{B}$ into account. A direct calculation leads to
\begin{eqnarray}\label{spinsumab}
\sum_h \lambda^{S}_h(p^\mu)\stackrel{\neg}{\lambda}^{S}_h(p^\mu)=m[\mathbbm{1}+\mathcal{G}(\phi)]\mathcal{A}, \\
\sum_h \lambda^{A}_h(p^\mu)\stackrel{\neg}{\lambda}^{A}_h(p^\mu)=-m[\mathbbm{1}-\mathcal{G}(\phi)]\mathcal{B}.
\end{eqnarray}

Now, in order to acquire Lorentz covariant spin sums, it could be imposed that $\mathcal{A}$ and $\mathcal{B}$ are simply the inverse of $[\mathbbm{1}+\mathcal{G}(\phi)]$ and $[\mathbbm{1}-\mathcal{G}(\phi)]$ respectively. However, $\det[\mathbbm{1}\pm \mathcal{G}(\phi)]=0$, and then this naive approach does not work \cite{vicinity}. An interesting general inverse was introduced in the literature by the outstanding works of Moore, and Penrose, which are summarized in \cite{baratamoore}. In general grounds these works give a complete algorithm to find out the so-called pseudo-inverse, hereafter denoted by $M^{+}$, of a singular matrix $M$. 
Looking towards finding a matrix that can be really classified as an inverse to the Lorentz break part of the spin sums, use was made of a ``$\tau-$deformation'' \cite{aaca,vicinity}, writing 
\begin{eqnarray}\label{spinsumab2}
\sum_h \lambda^{S}_h(p^\mu)\stackrel{\neg}{\lambda}^{S}_h(p^\mu)&=&m[\mathbbm{1}+\tau G(\phi)]\mathcal{A}\vert_{\tau\rightarrow 1},
\end{eqnarray} and 
\begin{eqnarray}\label{outra}
\sum_h \lambda^{A}_h(p^\mu)\stackrel{\neg}{\lambda}^{A}_h(p^\mu)&=&-m[\mathbbm{1}-\tau G(\phi)]\mathcal{B}\vert_{\tau\rightarrow 1}.
\end{eqnarray} 
In order to the $\tau-$deformation makes sense it must have a well defined $\tau\rightarrow 1$ limit. In addition, this limit is the unique necessary constraint used. The investigation of the vicinity of $[\mathbbm{1}+\mathcal{G}(\phi)]$ \cite{baratamoore} corroborates this approach. 

We start from the fact that both matrices $\mathbbm{1}$ and $\mathcal{G}(\phi)$ are non-singular. Once it is really true, there exist constant numbers $r_1$ and $r_2$, which, in turn, depend on $\mathbbm{1}$ and $\mathcal{G}(\phi)$, under the condition $0<r_1\leq r_2$, such that $\mathbbm{1}+\tau \mathcal{G}(\phi)$ is invertible for all $\tau\in\mathbb{C}$ with $0<\mid\tau\mid<r_1$ and for all $\tau\in\mathbb{C}$ with $\mid\tau\mid>r_2$. Hence it is possible to see that 
\begin{eqnarray}
\mathbbm{1}+\tau \mathcal{G}(\phi)=[\mathbbm{1}\tau+\mathcal{G}^{-1}(\phi)]\mathcal{G}(\phi),\label{qq1}
\end{eqnarray} in such a way that $[\mathbbm{1}+\tau \mathcal{G}(\phi)]$ is invertible only if $[\mathbbm{1}\tau+\mathcal{G}^{-1}(\phi)]$ is non-singular. Now take $Z\equiv -\mathcal{G}^{-1}(\phi)$ and $\{\upsilon_1,\ldots\upsilon_k\}\in\mathbb{C}$ to be the (not necessarily distinct) $k$ roots of the polynomial $p_z=\det[Z-\upsilon\mathbbm{1}]$. Notice from Eq. (\ref{qq1}) that 
\begin{eqnarray}
\mathbbm{1}+\tau \mathcal{G}(\phi)=-[Z-\mathbbm{1}\tau]\mathcal{G}(\phi).\label{qq2}
\end{eqnarray}
If all roots stand null, one should take $r_1=r_2$ as any positive number. Defining, then, $r_1\equiv\min\{|\upsilon_k|,\upsilon_k\neq0\}$ and $r_2\equiv\max\{|\upsilon_k|\}$, we have two distinct ranges of values for $\tau$ (exactly the ranges without roots of $p_z$) for which $\tau$ enables $[\mathbbm{1}\tau+\mathcal{G}^{-1}(\phi)]$ to be non-singular, ensuring $[\mathbbm{1}+\tau \mathcal{G}(\phi)]$ invertible. These ranges are $\{\tau\in\mathbb{R},0<|\tau|<r_1\}$ and $\{\tau\in\mathbb{R},|\tau|>r_2\}$. As it can be seen, 
\begin{equation}
p_z=\upsilon^4-2\upsilon^2+1 \label{qq3}
\end{equation} 
have roots given by $\pm 1$, both with multiplicity two. Therefore, one can see that the unique constraint is given by $|\tau|>1$ or $0<|\tau|<1$. Hence, the limit taken in \cite{aaca} shows to be indeed valid and the operators $\mathcal{A}$ and $\mathcal{B}$ may well be chosen in order to give the precise inverses of (\ref{spinsumab2}) and (\ref{outra}), respectively. 

Thus, now one is able to precisely write the $\mathcal{A}$ and $\mathcal{B}$ operators as
\begin{eqnarray}
\mathcal{A} = 2[\mathbbm{1}+\mathcal{G}(\phi)]^{-1} = 2\bigg[\frac{\mathbbm{1}-\tau\mathcal{G}(\phi)}{1-\tau^2}\bigg], 
\\
\mathcal{B} = 2[\mathbbm{1}-\mathcal{G}(\phi)]^{-1} = 2\bigg[\frac{\mathbbm{1}+\tau\mathcal{G}(\phi)}{1-\tau^2}\bigg], 
\end{eqnarray} 
noting that $\mathcal{G}(\phi)^2=\mathbbm{1}$, we define the following spin sums
\begin{eqnarray}
\label{spinsumsA2}\sum_{\alpha} \lambda^S_{\alpha}(p^{\mu})\stackrel{\neg}{\lambda}^{S}_{\alpha}(p^{\mu})&=& 2m[\mathbbm{1}+\mathcal{G}(\phi)]\bigg[\frac{\mathbbm{1}-\tau\mathcal{G}(\phi)}{1-\tau^2}\bigg]|_{\tau\rightarrow 1},\nonumber\\ 
&=&2m\mathbbm{1},
\\
\label{spinsumaB2}\sum_{\alpha} \lambda^A_{\alpha}(p^{\mu})\stackrel{\neg}{\lambda}^{A}_{\alpha}(p^{\mu})&=& -2m[\mathbbm{1}-\mathcal{G}(\phi)]\bigg[\frac{\mathbbm{1}+\tau\mathcal{G}(\phi)}{1-\tau^2}\bigg]\vert_{\tau\rightarrow 1},\nonumber\\
&=&-2m\mathbbm{1},
\end{eqnarray}
now, providing a Lorentz invariant outcome. 
Interesting enough, even after the adjoint redefinition, all the physical information encoded in the spinors remains unchanged.

\section{The underlying features about the Bilinear Analysis}\label{secaobilineares}

Let $\psi$ be a spinor field belonging to a section of the vector bundle $\mathbf{P}_{Spin^{e}_{1,3}}(\mathcal{M})\times\, _{\rho}\mathbb{C}^4$ where $\rho$ stands for the entire representation space $D^{(1/2,0)}\oplus D^{(0,1/2)}$, or a given sector of such. Pertti Lounesto defined the bilinear covariants associated to $\psi$ as \cite{crawford1,crawford2}
\begin{eqnarray}
\label{covariantes}
\sigma=\psi^{\dag}\gamma_{0}\psi, \hspace{1cm}  \omega=-\psi^{\dag}\gamma_{0}\gamma_{0123}\psi, \hspace{1cm} \mathbf{J}=\psi^{\dag}\mathrm{\gamma_{0}}\gamma_{\mu}\psi\;\theta^{\mu}, \nonumber\\
\mathbf{K}=\psi^{\dag}\mathrm{\gamma_{0}}\textit{i}\mathrm{\gamma_{0123}}\gamma_{\mu}\psi\;\theta^{\mu},\hspace{1cm} \mathbf{S}=\frac{1}{2}\psi^{\dag}\mathrm{\gamma_{0}}\textit{i}\gamma_{\mu\nu}\psi\theta^{\mu}\wedge\theta^{\nu}\,,
\end{eqnarray} 
the elements $\{ \theta^\mu \}$ are the dual basis of a given inertial frame $\{ \textbf{e}_\mu \} = \left\{ \frac{\partial}{\partial x^\mu} \right\}$, with $\{x^\mu\}$ being the global spacetime coordinates and the Dirac matrices are, in the Chiral (or Weyl) representation,
\begin{eqnarray}
\gamma_0 = \left( \begin{array}{cc}
\mathbb{O} & \mathbbm{1} \\ 
\mathbbm{1} & \mathbb{O}
\end{array}  \right), \qquad \gamma_i = \left( \begin{array}{cc}
\mathbb{O} & \sigma_{i} \\ 
-\sigma_{i} & \mathbb{O}
\end{array}  \right).
\end{eqnarray}

In general grounds, it is always expected to associate (\ref{covariantes}) to physical observables. For instance, in the usual case, bearing in mind the relativistic description of the electron, $\sigma$ is the invariant length, $\boldsymbol{J}$ is associated to the current density, $\boldsymbol{K}$ is the spin projection in the momentum direction, and $\boldsymbol{S}$ is the momentum electromagnetic density. Remarkably, in \cite{elkomonopole}, authors stated that a Majorana field can not carry any electric or magnetic charge, i.e., the physical currents $e_{M}\boldsymbol{J}_{M}$ and $q_{M}\boldsymbol{K}_{M}$ are null because $e_{M}=q_{M}=0$. This is a valid attempt to understand such quantities, however, it is not unique.

 The bilinear covariants, as well known, obey the so-called Fierz-Pauli-Kofink (FPK) identities, given by \cite{baylis}
\begin{eqnarray}\label{fpkidentidades}
\boldsymbol{J}^2 = \sigma^2+\omega^2, \quad J_{\mu}\!\!&K_{\nu}&\!\!-K_{\mu}J_{\nu} = -\omega S_{\mu\nu} - \frac{\sigma}{2}\epsilon_{\mu\nu\alpha\beta}S^{\alpha\beta}, \nonumber
\\
J_{\mu}K^{\mu} &=& 0, \quad \boldsymbol{J}^2 = -\boldsymbol{K}^2,
\end{eqnarray}
in which 
\begin{eqnarray}
\boldsymbol{J}^2 = (J_\mu \theta^\mu)(J^\nu  \textbf{e}_\nu) = J_\mu J^\mu,
\end{eqnarray}
where we have used the definition of the dual basis, $\theta^\mu(\textbf{e}_\nu)=\delta^\mu_\nu$, and similarly $\boldsymbol{K}^2 = K_\mu K^\mu$, both clearly being scalars.
It can be seen that the physical requirement of reality can always be satisfied for Dirac spinors bilinear covariants \cite{crawford1}, by a suitable deformation of the Clifford basis leading to physical appealing quantities.
The above identities are fundamental not only for classification, but also to further assert the inversion theorem\cite{crawford1}. Such identities are also valid for arbitrary dimensions and signatures, please check Ref \cite{fpkarbitrario}. Within the Lounesto classification scheme, a non vanishing $\mathbf{J}$ is crucial, since it enables to define the so called boomerang \cite{lounestolivro} which has an ample geometrical meaning to assert that there are precisely six different classes of spinors according to their bilinear covariants. In fact, this is a prominent consequence of the definition of a boomerang. As far as the boomerang is concerned, it is not possible to exhibit more than six types of spinors, according to the bilinear covariants. Indeed, Lounesto spinor classification splits regular and singular spinors. The regular spinors are those which have at least one of the bilinear covariants $\sigma$ and $\omega$ non-null. Singular spinors, on the other hand, have $\sigma = 0 =\omega$, consequently the Fierz identities are normally replaced by more general conditions \cite{crawford1,crawford2}
\begin{eqnarray}\label{multi}
Z^{2}=4\sigma Z, \nonumber\\ 
Z\gamma_{\mu}Z=4J_{\mu}Z,\nonumber\\ 
Zi\gamma_{5}Z=4\omega Z,\\
Zi\gamma_{\mu}\gamma_{\nu}Z=4S_{\mu\nu}Z,\nonumber\\
Z\gamma_{5}\gamma_{\mu}Z=4K_{\mu}Z\nonumber.
\end{eqnarray}
When an arbitrary spinor $\xi$ satisfies $\xi^{*}\psi\neq 0$ and belongs to $\mathbb{C}\otimes \mathcal{C}\ell_{1,3}$ — or equivalently when $\xi^{\dagger}\gamma_{0}\psi\neq 0 \in \mathcal{M}(4, \mathbb{C})$— it is possible to recover the original spinor $\psi$ from its aggregate $Z$. Such relation is given by $\psi = Z\xi$, and the aggregate reads
\begin{eqnarray}\label{multivetorz}
Z=\sigma+\mathbf{J}+i\mathbf{S}+\mathbf{K}\gamma_{5}-i\omega\gamma_{5},
\end{eqnarray}
and the spinor $\xi$ is given by the so-called Takahashi algorithm. Hence, using \eqref{multivetorz} and taking into account that we are dealing with singular spinors, it is straightforward to see that the aggregate can be recast as
\begin{eqnarray}
Z=\mathbf{J}(1+i\mathbf{s}+ih\gamma_{0123}), \label{ZB}
\end{eqnarray} 
where $\mathbf{s}$ is a space-like vector orthogonal to $\mathbf{J}$, and $h$ is a real number \cite{lounestolivro}. The multivector as expressed in \eqref{ZB} is a boomerang. By inspecting the condition $Z^{2}=4\sigma Z$ we see that for singular spinors $Z^2=0$. However, in order to the FPK identities to hold it is also necessary that both conditions $\mathbf{J}^2=0$ and $(\mathbf{s}+h\gamma_{0123})^2=-1$ must be satisfied \cite{spinorrepresentation}.

Unfortunately, the same cannot be stated for mass dimension one spinors, as Elko spinors. Actually, a straightforward calculation shows an incompatibility in the usual construction of bilinear covariants. In fact, one of the FPK identities is violated. The reason rests upon the new dual structure associated to these spinors. It is worth to mention that the main difference between the Crawford deformation \cite{crawford1,crawford2} and the one to be accomplished here is that in the former case, the spinors are understood as Dirac spinors, i. e., spinorial objects endowed with single helicity. Therefore, the dual structure is the usual one $\bar{\psi}(p^{\mu}) = \psi^{\dag}(p^{\mu})\gamma_0$ and the required normalization is also ordinary. On the other hand Elko spinors, due to their own formal structure, need their dual to be redefine $\stackrel{\neg}{\lambda}\;^{S/A}_{h}(p^{\mu}) = [\Xi(p^{\mu})\lambda^{S/A}_{h}(p^{\mu})]^{\dag}\gamma_0$. This redefinition leads, ultimately, to a new normalization culminating in a basis deformation satisfying the FPK identities.

\section{The physical information behind the Elko bilinear forms}\label{elkoforms}
Let us make these assertions more clear by explicitly showing the previously mentioned problem. We leave for the Appendix \ref{apendiceA} a brief and self contained overview on the spinorial formal structure of Lorentz breaking Elko fields. Taking advantage of what was there defined, we use as an example the spinor $\lambda^{S}_{\lbrace -,+\rbrace}(p^{\mu})$ and its dual given by, respectively,
\begin{eqnarray}\label{3}
\lambda^{S}_{\lbrace -,+\rbrace}(p^{\mu}) = \mathcal{B}_{-} 
\left(\begin{array}{c}
-i\sin(\theta/2) e^{-i\phi/2} \\ 
i\cos(\theta/2) e^{i\phi/2} \\  
\cos(\theta/2) e^{-i\phi/2} \\ 
\sin(\theta/2) e^{i\phi/2}
\end{array}  \right),
\end{eqnarray}   
and
\begin{eqnarray*}\label{14}
\stackrel{\neg}{\lambda}\;^{S}_{\lbrace -,+\rbrace}(p^{\mu}) &=& 
\\ 
&\mathcal{B}_{+}&\left( \begin{array}{cccc}
-i\sin(\theta/2) e^{i\phi/2} & i\cos(\theta/2) e^{-i\phi/2} & -\cos(\theta/2) e^{i\phi/2} & -\sin(\theta/2) e^{-i\phi/2}
\end{array}  \right)\!.
\end{eqnarray*} We reinforce once again that the dual structure associated to Elko spinors is obtained in a very judicious fashion \cite{aaca}, leaving no space to modifications, exception made to the generalizations found in Ref. \cite{jcap,aaca}. Using (\ref{3}) and (\ref{14}), as a direct calculation shows, Eqs. (\ref{covariantes}) give 
\begin{eqnarray}
\sigma &=& -2m,\label{cov1}\\
\omega &=& 0,\label{cov2}\\
J_0 &=& 0, \\
J_1 &=& 2im\cos\theta\cos\phi, \\
J_2 &=& 2im\cos\theta\sin\phi,\\
J_3 &=& -2im\sin\theta, \label{cov3}\\
K_0 &=& 0,\\
K_1 &=& -2m\sin\phi,\\
K_2 &=& 2m\cos\phi,\\
K_3 &=& 0,\label{cov4}
\end{eqnarray}
and 
\begin{eqnarray}
S_{01} &=& -2im\sin\theta\cos\phi,\\
S_{02} &=& -2im\sin\theta\sin\phi,\\
S_{03} &=& -2im\cos\theta,\\
S_{12} &=& S_{13} = S_{23} = 0.\label{cov5}
\end{eqnarray}
As it can be verified, the above bilinear covariants do not respect the FPK quadratic relations. More specifically, the relation containing $S_{\mu\nu}$ is not fulfilled. In view of this problem, we revisited the formulation developed in \cite{crawford1} looking towards to find out an appropriated Clifford basis upon which the bilinear covariants can be constructed, leading to the right observance of the FPK relations. We highlight that the price to be paid is that only a subset of bilinear covariants comprises real quantities.

\section{A path to the Clifford algebra basis deformation: An alternative method to ensure the FPK relations}\label{secaofpk}

As it is well known, the constitutive relation of the Clifford algebra is given by 
\begin{eqnarray}\label{eita}
\lbrace \gamma_{\mu}, \gamma_{\nu}\rbrace = 2g_{\mu\nu}\mathbbm{1}, \quad \mu,\nu = 0,1,2,...,N-1,
\end{eqnarray}
where $g_{\mu\nu}$ is a $N=2n$ even-dimensional space-time metric, which in Cartesian coordinates has the form $diag(1,-1,...,-1)$. The rudimentary generators of the Clifford algebra are, then, the identity $\mathbbm{1}$ and the vectors $\gamma_\mu$, usually represented as square matrices. The standard approach dictates the complementation of the Clifford algebra basis, in order to ensure real bilinear forms. This complement is performed by the composition of the vector basis, used as building blocks \cite{crawford1} 
\begin{eqnarray}\label{gammatil}
\tilde{\gamma}_{\mu_{1}\mu_{2}...\mu_{N-M}} \equiv \frac{1}{M!}\epsilon_{\mu_{1}\mu_{2}...\mu_{N}}\gamma^{\mu_{N-M+1}}\gamma^{\mu_{N-M+2...}}\gamma^{\mu_{N}}.
\end{eqnarray} 
As it is easy to see, the lowest $M$ value is two (the smallest combination), nevertheless, it runs in the range $M=2,3,...,N$. In this respect, the elements that form the (real) Clifford algebra basis are 
\begin{eqnarray}\label{set}
\lbrace \Gamma_{i}\rbrace \equiv \lbrace \mathbbm{1}, \gm, \tilde{\gamma}_{\mu_{1}\mu_{2}...\mu_{N-2}},..., \tilde{\gamma}_{\mu},\tilde{\gamma}\rbrace,  
\end{eqnarray} where $\tilde{\gamma}\equiv \tilde{\gamma}_{\mu_{1}\mu_{2}...\mu_{N-N}}$.    

In view of the new elements appearing in the definition of the Elko dual, it is necessary to adequate the Clifford algebra basis complementation. As shown previously, it is absolutely necessary for the right appreciation of the FPK relations. We shall stress that for the Dirac spinorial case, the set (\ref{set}) is suitable deformed (by a slightly different normalization) in order to provide real bilinear covariants. We shall pursue something similar here, and we are successful in correction the problem related to the FPK relations. Notwithstanding, only a subset of bilinear covariants ends up real in the Elko spinorial case.  

The first two bilinear arising from the Clifford algebra basis are 
\begin{eqnarray}
\sigma &\equiv&\lambdabp\mathbbm{1}\lambdap,\label{sig}\\
J_{\mu} &\equiv&\lambdabp\gm\lambdap\label{jota},
\end{eqnarray}
where $\lambdabp =  [\Xi(p^{\mu})\lambdap]^{\dag}\gamma_0$ is identified as the Elko dual spinor. The operator $\Xi(p^{\mu})$ is responsible to change the spinor helicity, here labelled by $h$. The requirement $\sigma=\sigma^\dagger$ automatically leads to $\gamma_0 = \Xi^{\dag}(p^{\mu})\gamma_0^{\dag}\Xi(p^{\mu})$ since $\Xi^{2}(p^{\mu})=\mathbbm{1}$. This constraint is readily satisfied, in such a way that (\ref{sig}) ensures a real quantity. By the same akin reasoning, one should impose $J_{\mu} = J_{\mu}^{\dag}$. This requirement leads to the following constraint $\gamma_0\gm = \Xi^{\dag}(p^{\mu})\gm^{\dag}\gamma_0^{\dag}\Xi(p^{\mu})$ which, however, cannot be fulfilled in any case. In fact, the counterpart associated to Dirac spinors is simple and readily satisfy $\gamma_0^{-1}\gamma_{\mu}^\dagger\gamma_0=\gamma_\mu$, a constraint naturally achieved by construction. 

The Elko's dual structure has forced a new interpretation of the bilinear covariants. Note that, firstly, for mass dimension one spinors $J^\mu$ cannot be associated to the conserved current. Obviously, in order to have $\partial_\mu J^\mu=0$ it is sufficient to use the Dirac equation. This truism has lead to interesting algebraic possibilities \cite{cjr}, but the point to be emphasized here is that there is no problem in having a complex quantity related to the bilinear $J^\mu$. The additional important consequence of an imaginary $J^\mu$ is that in order to satisfy the FPK equations one must have $K^\mu$ or $S^{\mu\nu}$ also imaginary.  

Notice that, instead of the usual Crawford deformation, here we do not arrive at an entire real bilinear set. In fact, in trying to implement the full reality condition it is mandatory to change the building block of the Clifford basis $\gamma_\mu$. It would inevitably lead, however, to a change in the constitutive algebraic relation of the Clifford algebra (\ref{eita}). Therefore, this change must be excluded. It is important to emphasize, moreover, that even being willing to accept a modification of (\ref{eita}) the resulting constraint to get a real set of bilinears cannot be fulfilled.

Having said that, we may proceed deforming the usual basis in order to redefine bilinear covariants which satisfy the FPK identities. Making use of Eq. \eqref{gammatil} and considering that the norm for the Elko spinor is real we have 
\begin{eqnarray*}
[\lambdabp\tilde{\gamma}_{\mu_{1}\mu_{2}...\mu_{N-M}}\lambdap]^{\dag}= (-1)^{M(M-1)/2}\lambdabp\Xi(p^{\mu})\tilde{\gamma}_{\mu_{1}\mu_{2}...\mu_{N-M}}\Xi(p^{\mu})\lambdap. \label{urru}
\end{eqnarray*} It can be readily verified that the following redefinition of $\tilde{\gamma}_{\mu_{1}\mu_{2}...\mu_{N-M}}$ is appropriate to ensure $K^\mu$ as a real quantity:
\begin{eqnarray}\label{gammatil2}
\tilde{\gamma}_{\mu_{1}\mu_{2}...\mu_{N-M}}= (i^{M(M-1)/2}/M!)\Xi(p^{\mu})\epsilon_{\mu_{1}...\mu_{N}}\gamma^{\mu_{N-M+1}...}\gamma^{\mu_{N}}\Xi(p^{\mu}).
\end{eqnarray}
With the redefinition above, one is able to define the bispinor Clifford algebra basis as in (\ref{set}), but with the gammas given by (\ref{gammatil2}). As an example, consider the four-dimensional space-time. In this case the basis is given, accordingly, by 
\begin{eqnarray}
M &=& 4 \quad\Rightarrow\quad \tilde{\gamma} = -i \Xi(p^{\mu})\gamma_5\Xi(p^{\mu}), \\
M &=& 3 \quad\Rightarrow\quad \tilde{\gamma}_{\mu} = -\Xi(p^{\mu})\gamma_5\gamma_{\mu}\Xi(p^{\mu}),\\
M &=& 2 \quad\Rightarrow\quad  \tilde{\gamma}_{\mu\nu} = \frac{i}{2}\Xi(p^{\mu}) \gamma_{\mu}\gamma_{\nu} \Xi(p^{\mu}),
\end{eqnarray}
where $\gamma_5=-i\gamma_0\gamma_1\gamma_2\gamma_3$. Now, with the real Clifford algebra basis at hands, it is possible to construct the bilinear forms, given by
\begin{eqnarray}\label{Elkonewbilinear}
\mathbbm{1}&\quad\Rightarrow\quad&\sigma_{E} = \lambdabp\mathbbm{1}\lambdap, \nonumber \\
\gamma_{\mu}&\quad\Rightarrow\quad& J_{\mu_{E}}= \lambdabp\gamma_{\mu}\lambdap,\nonumber\\
\tilde{\gamma}&\quad\Rightarrow\quad& \omega_{E} = -i\lambdabp\gamma_5\lambdap,\\
\tilde{\gamma}_{\mu}&\quad\Rightarrow\quad& K_{\mu_{E}}= -\lambdabp\Xi(p^{\mu})\gamma_5\gamma_{\mu}\Xi(p^{\mu})\lambdap,\nonumber\\
\tilde{\gamma}_{\mu\nu}&\quad\Rightarrow\quad& S_{\mu\nu_{E}} = i\lambdabp \Xi(p^{\mu})\gamma_{\mu}\gamma_{\nu}\Xi(p^{\mu})\lambdap.\nonumber
\end{eqnarray}
From above construction one ensures that the slight modifications of the bilinear covariants are enough to ensure the right observance of the FPK identities \eqref{fpkidentidades}. After all, we arrive at $\sigma$ and $K^\mu$ as real and non-null quantities. 

A hint towards the interpretation of the bilinear forms is developed in \cite{beyondlounesto}, and we remark that such interpretation can be extended for any general spinor (charged and uncharged - independent of the spinor and its associated adjoint structure). Thus, $\sigma$ still stands for the invariant length (even known as mass term). Moreover, the four-vector $\textbf{J}$ may represent the four-velocity \cite{lounestolivro} as well as the electric current density for charged particles, whereas for neutral particles we are led to interpret it as remarked in \cite{giuntineutrino} the effective electromagnetic current four-vector. The bilinear $\textbf{K}$ shall be related with the spin alignment due to a coupling with matter or electromagnetic field. Lastly, $\textbf{S}$ is related to the electromagnetic moment density for charged particles, although for neutral particles we can infer that such physical quantity may correspond to the momentum spin-density, or, even steeped in the commentary woven in \cite{studenikinneutrino2}, it may represent spin precession (spin oscillation) in the presence of matter or electromagnetic fields. Such effect is caused by the neutral particle interaction with matter polarized by external magnetic field or, equivalently, by the interaction of the induced magnetic moment of a neutral particle with the magnetic field \cite{grigoneutrino,ahluwalianeutrino}. The meaning of the $\omega$ bilinear is, in this very moment,  not clear enough for us to get a physical interpretation.

\subsection{Investigation of the Elko's covariant structure}\label{secaocovariancia}

So far we have worked out quantities defined as (\ref{Elkonewbilinear}) claiming that they must be faced as bilinear covariants. While they are bilinear quantities, their covariant structure must be evinced. All the issue is related to the (necessary) presence of the $\Xi(p^{\mu})$ operator. 

Supposing that the Elko spinor belongs to a linear representation of the symmetry group in question, in such a way that seen by another frame the field undergoes a transformation as
\begin{eqnarray}\label{tl}
\lambda^{\prime S/A}_{h}(p^{\mu\prime}) = S(\Lambda) \lambda_{h}^{S/A}(p^{\mu}).
\end{eqnarray} There is a Dirac-like operator that annihilates Elko \cite{aaca} (not related to the field dynamics) given by
\begin{eqnarray}\label{eqdiracElko}
\big(\gamma_{\mu}p^{\mu}\Xi(p^{\mu}) \pm m\big)\lambda^{S/A}_{h}(p^{\mu})=0, 
\end{eqnarray} from which we shall investigate the covariance. Applying the transformation \eqref{tl} for the fields in the equation \eqref{eqdiracElko}, we find 
\begin{eqnarray}
\big(\gamma_{\mu}p^{\prime\mu}\Xi(p^{\mu}) \pm m\big)\lambda^{\prime S/A}_{h}(p^{\mu\prime})=0.
\end{eqnarray}  
The momentum can be written as $p^{\mu}\leftrightarrow i\partial^{\mu}$ and the partial derivative transforms usually as $\partial^{\prime\mu} = \Lambda^{\mu}_{\;\;\;\beta}\partial^{\beta}$. Therefore, in order to ensure covariance of Eq. (\ref{eqdiracElko}) it is necessary the following behavior of the Dirac matrices and the $\Xi(p^{\mu})$ operator, respectively
\begin{eqnarray}
 \gamma^{\prime}_{\beta} &=& S(\Lambda)\gamma_{\mu}S^{-1}(\Lambda)\Lambda^{\mu}_{\;\;\;\beta},\label{gamma1}\\
 \Xi^{\prime}(p^{\prime\mu}) &=& S(\Lambda)\Xi(p^{\mu}) S^{-1}(\Lambda).\label{xi1}
\end{eqnarray} Equation (\ref{gamma1}) is the usual requirement to be inputed to the gamma matrices in order to achieve a covariant Dirac equation. The requirement (\ref{xi1}) is the new ingredient of the Elko theory, which must be investigated. 

Interestingly enough, from the expression \eqref{XI}, along with \eqref{gamma1}, it is possible to see that \cite{aaca,speranca}
\begin{eqnarray*}
\Xi^{\prime}(p^{\prime\mu}) &=&\frac{1}{2m}\Big( \lambda^{\prime S}_{\lbrace +-\rbrace}(p^{\mu\prime})\bar{\lambda}^{\prime S}_{\lbrace +-\rbrace}(p^{\mu\prime}) + \lambda^{\prime S}_{\lbrace -+\rbrace}(p^{\mu\prime})\bar{\lambda}^{\prime S}_{\lbrace -+\rbrace}(p^{\mu\prime}) \\
&&- \lambda^{\prime A}_{\lbrace +-\rbrace}(p^{\mu\prime})\bar{\lambda}^{\prime A}_{\lbrace +-\rbrace}(p^{\mu\prime})-\lambda^{\prime A}_{\lbrace -+\rbrace}(p^{\mu\prime})\bar{\lambda}^{\prime A}_{\lbrace -+\rbrace}(p^{\mu\prime}) \Big) \\
&=& \frac{1}{2m}S(\Lambda)\Big( \lambda^{S}_{\lbrace +-\rbrace}(p^{\mu})\bar{\lambda}^{S}_{\lbrace +-\rbrace}(p^{\mu}) + \lambda^{S}_{\lbrace -+\rbrace}(p^{\mu})\bar{\lambda}^{S}_{\lbrace -+\rbrace}(p^{\mu}) 
\nonumber\\
&&- \lambda^{A}_{\lbrace +-\rbrace}(p^{\mu})\bar{\lambda}^{A}_{\lbrace +-\rbrace}(p^{\mu})-\lambda^{A}_{\lbrace -+\rbrace}(p^{\mu})\bar{\lambda}^{A}_{\lbrace -+\rbrace}(p^{\mu}) \Big)\gamma_0 S^{\dag}(\Lambda)\gamma_0,
\end{eqnarray*} and therefore 
\begin{eqnarray}
\Xi^{\prime}(p^{\prime\mu}) = S(\Lambda)\Xi(p^{\mu}) S^{-1}(\Lambda),
\end{eqnarray} 
as expected. Once verified the right transformations, we are able to evince the bilinear quantities. Starting from $\sigma$, we have
\begin{eqnarray*}
\sigma_{E}^{\prime} &=& \stackrel{\neg}{\lambda}\;^{\prime S/A}_{h}(p^{\mu\prime})\lambda^{\prime S/A}_{h}(p^{\mu\prime})\\
&=&\lambda^{\dag S/A}_h(p^{\mu}) S^{\dag}(\Lambda)S^{-1\dag}(\Lambda)\Xi^{\dag}(p^{\mu}) S^{\dag}(\Lambda)\gamma_0 S(\Lambda)\lambda^{S/A}_h(p^{\mu}) \\
&=&  \stackrel{\neg}{\lambda}\;^{S/A}_{h}(p^{\mu})\lambda^{S/A}_{h}(p^{\mu})\\
&=&  \sigma_{E},
\end{eqnarray*} 
implying $\sigma$ to be a scalar. Repeating the same procedure for the remaining bilinear forms, we obtain 
\begin{eqnarray*}
J^{\prime}_{\mu_{E}}&\rightarrow & \Lambda^{\nu}_{\mu}\lambdabp\gamma_{\nu}\lambdap, \quad \mbox{(Vector)},\\
\omega^{\prime}_{E} &\rightarrow &  det(\Lambda)i\lambdabp\Xi(p^{\mu})\gamma_5\Xi(p^{\mu})\lambdap, \quad  \mbox{(Scalar)},\\
K^{\prime}_{\mu_{E}}&\rightarrow & -det(\Lambda)\Lambda^{\nu}_{\;\;\;\mu}i\lambdabp\Xi(p^{\mu})\gamma_5\gamma_{\nu}\Xi(p^{\mu})\lambdap, \quad \mbox{(Vector)},\\
S^{\prime}_{\mu\nu_{E}} &\rightarrow & \frac{i}{2} \Lambda^{\alpha}_{\;\;\;\mu} \Lambda^{\beta}_{\;\;\;\nu}\lambdabp \Xi(p^{\mu})\gamma_{\alpha}\gamma_{\beta}\Xi(p^{\mu})\lambdap. \quad \mbox{(Bivector).}
\end{eqnarray*} 
Therefore, the nomenclature previously adopted is indeed adequate to the case at hands. We shall finalize pointing out that the investigation of covariance in which concerns $SIM(2)$, $HOM(2)$ Lorentz subgroups is taken into account in Appendix \ref{apendiceA}. The analysis is quite analogous, and the physical statements essentially hold the same.


\section{Looking towards some applications: The non-linear Heisenberg theory formalism}\label{secaoHeisenberg}
The non-linear Heisenberg equation of motion is easily obtained by varying the action with respect to the spinor field, constructed by \cite{heisenbergpaper,novello2}
\begin{eqnarray}
\mathcal{L} = \frac{i}{2}\bar{\psi}^{H}\gamma^{\mu}\partial_{\mu}\psi^{H} - \frac{i}{2}\partial_{\mu}\bar{\psi}^{H}\gamma^{\mu}\psi^{H}-s J_{\mu}J^{\mu},
\end{eqnarray}
where $J_{\mu}$ is defined in \eqref{covariantes}, but now replacing $\psi$ by $\psi^H$. Thus, non-linear Heisenberg equation reads\footnote{The fundamental field equations must be non-linear in order to represent interaction. The masses of the particles should be a consequence of this interaction \cite{heisenbergpaper1}.} \cite{novello}
\begin{eqnarray}\label{eqheisenberg}
i\gamma^{\mu}\partial_{\mu}\psi^{H} - 2s(A+iB\gamma^5)\psi^H=0,
\end{eqnarray}
where $s$ stands for a constant which has dimension of $(length)^2$ and the physical quantities $A$ and $B$ are given in terms of the usual bilinear covariants associated with the Heisenberg spinor, given by $A=\bar{\psi}^{H}\psi^{H}$ and $B=i\bar{\psi}^{H}\gamma_5\psi^{H}$, respectively.
The Heisenberg spinor can be represented by a line in a two-dimensional plane ($\pi$), where each axis is represented by the left-hand and right-hand spinors \cite{novello}. In such a way that, the Heisenberg spinor can be portrayed as the following identity
\begin{eqnarray}\label{heisenbergdecomposto}
\psi^H &=& \psi^{H}_{L}+ \psi^{H}_{R},
\end{eqnarray}
in other words,
\begin{eqnarray}
\psi^H=\frac{1}{2}(\mathbbm{1}+\gamma^5)\psi^H+\frac{1}{2}(\mathbbm{1}-\gamma^5)\psi^H.
\end{eqnarray}

A particular class of solutions of the Heisenberg equation \eqref{eqheisenberg} is given by  
\begin{eqnarray}\label{inomatacond}
\partial_{\mu}\psi = (aJ_{\mu}+bK_{\mu}\gamma^5)\psi,
\end{eqnarray}
with $a$, $b$ $\in \mathbb{C}$ having dimensionality $(length)^2$, and $J_{\mu}$ and $K_{\mu}$ being covariant and irrotational currents.
A solution of Eq. \eqref{inomatacond} also satisfies the Heisenberg equation of motion if $a$ and $b$ are such that $2s=i(a-b)$ \cite{novello} and shall be called as RIM (restricted Inomata-McKinley) spinor. As shown in \cite{RIM} every Dirac spinor written in terms of RIM spinors belongs to the class $1$ within Lounesto Classification. In order to Eq. \eqref{inomatacond} be integrable, the constants $a$ and $b$ must obey the constraint $Re(a)=Re(b)$.

Hence, we are able to define $J^2=J_{\mu}J^{\mu}$ and consequently 
\begin{eqnarray}
J_{\mu}=\partial_{\mu}S, 
\end{eqnarray}
where 
\begin{eqnarray}\label{eqS}
S=\frac{1}{(a+\bar{a})}\ln \sqrt{J^2},
\end{eqnarray}
represents a scalar, and similarly we can write 
\begin{eqnarray}\label{eqK}
K_{\mu}=\partial_{\mu}R,
\end{eqnarray}
with 
\begin{eqnarray}\label{eqR}
R=\frac{1}{(b-\bar{b})}\ln\bigg(\frac{A-iB}{\sqrt{J^2}}\bigg),
\end{eqnarray}
also being a scalar\footnote{In order to make the notation compact, we define $\sqrt{J^2}\equiv J$.}. 
From \eqref{inomatacond}, we obtain for the left-hand and right-hand Heisenberg spinors
\begin{eqnarray}
&&\partial_{\mu}\psi^{H}_{L} = (aJ_{\mu}+bK_{\mu})\psi^H_L, \\
&&\partial_{\mu}\psi^{H}_{R} = (aJ_{\mu}-bK_{\mu})\psi^H_R.
\end{eqnarray}
Thus, to complete the program we write an arbitrary spinor field, $\psi$, in terms of a decomposed Heisenberg spinor 
\begin{eqnarray}\label{psiarbitrario}
\psi=e^{F}\psi^H_L+e^{G}\psi^H_R,
\end{eqnarray}  
and then, looking towards to write a linear theory in terms of a non-linear theory, one analyses the properties encoded on the functions $F$ and $G$ in order to the spinor \eqref{psiarbitrario} to satisfy the Dirac equation. This is the prescription used in reference \cite{novello} to write Dirac spinors in terms of RIM-spinors. We will follow this idea in the next Section in order to also write Elko spinors in terms of RIM-spinors.

\section{Elko and RIM-spinors}\label{MDORIM}
Analogously as developed in \cite{novello}, we analyse the possibility to write an Elko fermionic field \cite{aaca} in terms of the non-linear Heisenberg spinors. All the discussion is based on two fundamental equations, the non-linear Heisenberg equation and the \emph{Dirac-like} equation for Elko fermions \cite{jcap}, which reads
\begin{eqnarray}\label{diraclike}
(i\gamma^{\mu}\partial_{\mu}\Xi(p^{\mu})\pm m\mathbbm{1})\lambda_{h}^{S/A}(\boldsymbol{x})=0, 
\end{eqnarray}
where the subscript $h$ stands for the helicity $h=\lbrace\pm,\mp\rbrace$, the upperindex $S/A$ stands for the self-conjugated and anti-self-conjugated spinors, respectively, under action of the charge conjugation operator ($\mathcal{C}\lambda_{h}^{S}= + \lambda_{h}^{S}$ and $\mathcal{C}\lambda_{h}^{A}= - \lambda_{h}^{A}$), and the operator $\Xi(p^{\mu})$ is given in \eqref{xioperator} \cite{bilineares}.
Then, we obtain the identity
\begin{eqnarray}
\lambda^{S/A}_h = \frac{1}{2}(\mathbbm{1}+\gamma^5)\lambda^{S/A}_{h}+\frac{1}{2}(\mathbbm{1}-\gamma^5)\lambda^{S/A}_{h}.
\end{eqnarray}
 One can explicit the left- and right-handed components as
\begin{eqnarray}
&&\lambda^{S/A}_{R_h}=\frac{1}{2}(\mathbbm{1}-\gamma_5)\lambda^{S/A}_{h}, \\
&&\lambda^{S/A}_{L_h}=\frac{1}{2}(\mathbbm{1}+\gamma_5)\lambda^{S/A}_{h}.
\end{eqnarray} 
We are now able to initiate the process to reach the decomposition (or representation) of the Elko spinors in terms of RIM-spinors. Firstly, one can write
\begin{eqnarray}
\lambda^{S/A}_{h} = e^{\stackrel{\neg}{F}}\psi_{L_{h^{\prime}}}^{H}+ e^{\stackrel{\neg}{G}}\psi_{R_{h^{\prime}}}^{H},
\end{eqnarray}
and, consequently, for the left- and right-handed components, we obtain
\begin{eqnarray}
&&\lambda^{S/A}_{L_h}= e^{\stackrel{\neg}{F}}\psi_{L_{h^{\prime}}}^{H},\\
&&\lambda^{S/A}_{R_h}= e^{\stackrel{\neg}{G}}\psi_{R_{h^{\prime}}}^{H}. 
\end{eqnarray}
The symbol ``\;$\stackrel{\neg}{}$\;" over $F$ and $G$, although commonly used to represent the dual of $\lambda$, is here simply to denote the functions related to $\lambda$ in the attempt to RIM-decompose such a spinor, and do not have any relation to the dual structure.
From the algebraic point of view, the sum of the spinors of certain classes did not prescribe the class, that is to say the resulting spinor does not necessarily belongs to the same class of spinors of which it is composed.

Following the program, the next step is to find the explicit form of $\stackrel{\neg}{F}$ and $\stackrel{\neg}{G}$ in order that $\lambda^{S/A}_{h}$ satisfies \eqref{diraclike}. By the same akin reasoning presented in \cite{novello} but now for the Elko spinors, we note that 
\begin{eqnarray}\label{chain}
\partial_{\mu} = \partial_{\mu}S\frac{\partial}{\partial S} + \partial_{\mu}R\frac{\partial}{\partial R}.
\end{eqnarray}
Taking into account the relations in equations \eqref{eqS} and \eqref{eqR}, we are able to write \eqref{chain} in this fashion
\begin{eqnarray}
\partial_{\mu}= J_{\mu}\frac{\partial}{\partial S} + K_{\mu}\frac{\partial}{\partial R},
\end{eqnarray}
therefore, one obtains
\begin{eqnarray}\label{elkomaoesq}
\partial_{\mu} \lambda^{S/A}_{L_h} = \bigg(\frac{\partial \stackrel{\neg}{F}}{\partial S}J_{\mu}+ \frac{\partial \stackrel{\neg}{F}}{\partial R}K_{\mu}\bigg)\lambda^{S/A}_{L_h}+(aJ_{\mu}+bK_{\mu})\lambda^{S/A}_{L_h}, 
\end{eqnarray}
and
\begin{eqnarray}\label{elkomaodir}
\partial_{\mu} \lambda^{S/A}_{R_h} = \bigg(\frac{\partial \stackrel{\neg}{G}}{\partial S}J_{\mu}+ \frac{\partial \stackrel{\neg}{G}}{\partial R}K_{\mu}\bigg)\lambda^{S/A}_{R_h}+(aJ_{\mu}-bK_{\mu})\lambda^{S/A}_{R_h}.
\end{eqnarray}
Taking advantage of the \emph{Dirac-like} equation, we multiply the equations \eqref{elkomaoesq} and \eqref{elkomaodir} by $i\gamma^{\mu}$, then, using the fact that $\Xi^2(p^{\mu})=\mathbbm{1}$ and $[\Xi(p^{\mu}),\gamma^{\mu}p_{\mu}]=0$, so we have\footnote{The authors choose to work in abstract only with the $\lambda^{S}_{h}$ spinors since the physical content holds the same for all the other Elko spinors, one differing from the other only by a constant phase.} 
\begin{eqnarray}
i\gamma^{\mu}\partial_{\mu}\lambda^{S}_{h} &=& i(A-iB)\bigg(\frac{\partial \stackrel{\neg}{F}}{\partial S}- \frac{\partial \stackrel{\neg}{F}}{\partial R} +(a-b)\bigg)\lambda^{S}_{R_h}+ i(A+iB)\bigg(\frac{\partial \stackrel{\neg}{G}}{\partial S}- \frac{\partial \stackrel{\neg}{G}}{\partial R} +(a-b)\bigg)\lambda^{S}_{L_h}\nonumber\\
&=&m\Xi(p^{\mu})\lambda^{S}_{h}.
\end{eqnarray}
Using the relations \eqref{A5}-\eqref{A8}, one obtains the following set of equations:
\begin{eqnarray}
&&\bigg[(A-iB)\bigg(\frac{\partial \stackrel{\neg}{F}}{\partial S}- \frac{\partial \stackrel{\neg}{F}}{\partial R} +(a-b)\bigg)\mathbbm{1}+im\Xi_1(p^{\mu})\bigg]\lambda^{S}_{R_h}=0, \\
&&\bigg[(A+iB)\bigg(\frac{\partial \stackrel{\neg}{G}}{\partial S}+ \frac{\partial \stackrel{\neg}{G}}{\partial R} +(a-b)\bigg)\mathbbm{1}+im\Xi_2(p^{\mu})\bigg]\lambda^{S}_{L_h}=0.
\end{eqnarray}
At this stage, we freely summarized the notation and rewrite \eqref{xioperator} as it follows
\begin{eqnarray}
\Xi(p^{\mu})= \left( \begin{array}{cc}
\Xi_1(p^{\mu})  & 0_{2\times 2} \\ 
0_{2\times 2} & \Xi_2(p^{\mu}) 
\end{array} \right).
\end{eqnarray}
After a bit of straightforward calculation, the solutions for $\stackrel{\neg}{F}(S,R)$ and $\stackrel{\neg}{G}(S,R)$ functions are given by 
\begin{eqnarray}
 \stackrel{\neg}{F}_{\pm}(S,R) & \equiv & -2isR \pm \frac{p\sin{\theta} (A+iB) e^{-2(a+\bar{a})S}}{2(a+\bar{a})},\\
 \stackrel{\neg}{G}_{\pm}(S,R) & \equiv & +2isR \pm \frac{p\sin{\theta} (A-iB) e^{-2(a+\bar{a})S}}{2(a+\bar{a})}.
\end{eqnarray}
Note that 
\begin{equation}
 A+iB = \frac{J^2}{A-iB},
\end{equation}
and from \eqref{eqS}, we have
\begin{equation}
J^{2} = e^{2(a+\bar{a})S}.
\end{equation}
Therefore,
\begin{eqnarray}
 e^{\stackrel{\neg}{F}_{-}} &=& \exp{\left[ -2isR - \frac{1}{2}\frac{p\sin{\theta}}{(a+\bar{a})(A-iB)} \right] },\\
 e^{\stackrel{\neg}{G}_{-}} &=& \exp{\left[ +2isR - \frac{1}{2}\frac{p\sin{\theta}}{(a+\bar{a})(A+iB)} \right] }.
\end{eqnarray}
By defining $\vartheta \equiv e^{2isR}$, we have
\begin{eqnarray}
 e^{\stackrel{\neg}{F}_{-}} &=& \frac{1}{\vartheta} \exp{\left[ - \frac{1}{2}\frac{p\sin{\theta}}{(a+\bar{a})(A-iB)} \right]},\\
 e^{\stackrel{\neg}{G}_{-}} &=& \vartheta \exp{\left[ - \frac{1}{2}\frac{p\sin{\theta}}{(a+\bar{a})(A+iB)} \right] }.
\end{eqnarray}
Following an analogue prescription, we can write
\begin{eqnarray}
 e^{\stackrel{\neg}{F}_{\pm}} & = & \frac{1}{\vartheta} \exp{\left[ \pm \frac{1}{2}\frac{p\sin{\theta}}{(a+\bar{a})(A-iB)} \right] },\\
 e^{\stackrel{\neg}{G}_{\pm}} & = & \vartheta \exp{\left[ \pm \frac{1}{2}\frac{p\sin{\theta}}{(a+\bar{a})(A+iB)} \right] }.
\end{eqnarray}
In this manner, we finally write the Elko spinors in terms of RIM-spinors:
\begin{equation}
 \lambda_h = \frac{1}{\vartheta} \exp{\left[ \pm \frac{p\sin{\theta}}{2(a+\bar{a})(A-iB)} \right]} \psi^H_{L_h} + \vartheta \exp{\left[ \pm \frac{p\sin{\theta}}{2(a+\bar{a})(A+iB)} \right] }\psi^H_{R_h},
\end{equation}
or, one is able to write the last expression in the fashion (replacing $p$ to $m$)
\begin{eqnarray}\label{elkorim}
\lambda = \bigg(\sqrt{\frac{J}{A-iB}}\bigg)^{\rho} \exp{\left[ \pm \frac{m\sin{\theta}}{4Re(a)(A-iB)} \right]} \psi^H_L + \bigg(\sqrt{\frac{A-iB}{J}}\bigg)^{\rho} \exp{\left[ \pm \frac{m\sin{\theta}}{4Re(a)(A+iB)} \right] }\psi^H_R,\nonumber\\
\end{eqnarray}
where we have defined 
\begin{eqnarray}
\rho & \equiv & \frac{Im(a)-Im(b)}{Im(b)} \nonumber\\ 
&=& \frac{-2s}{\text{Im}(b)}.
\end{eqnarray}
Note that we omitted the upper index $S/A$ due to the fact that such spinors differs from a global phase. As a net result we reach that the Elko fields can be freely represented as a combination of RIM-spinors which satisfy the non-linear Heisenberg equation.

\subsection{Two-dimensional spinor-spaces: the spinor-plane}\label{spinorplane}

We start this Section giving the definition \cite{restrictedinomata} of the spaces in which we will work on.

\begin{definition}\label{def1}
 We denote by $\Pi^H$ the two-dimensional space whose the set $\mathcal{B} = \{ \Psi^H_L, \Psi^H_R \}$ (namely, the left- and right-handed components of the RIM-spinor $\Psi^H$) forms a basis. Analogously, we denote the spaces $\Pi^D$ (with basis $\mathcal{D} = \{\Psi^D_L,\Psi^D_R\}$ being formed by the components of the Dirac-RIM spinor) and $\Pi^M$ (with basis formed by the Elko-RIM components $\mathcal{M} = \{\lambda_L,\lambda_R\}$). These spaces will be called spinor-planes.
\end{definition}

In order to achieve better organization, let us record that we can write Dirac spinors \cite{novello} $\Psi^D$ and Elko spinors $\lambda$ in the $\Pi^H$ space, via basis $\mathcal{B}$, as \cite{restrictedinomata}
\begin{eqnarray}
 \Psi^D & = & \exp{\left[\frac{iM}{(a+\bar{a})J}\right]} J^{2\sigma} \left(\sqrt{\frac{J}{A-iB}}\Psi^H_L + \sqrt{\frac{A-iB}{J}}\Psi^H_R \right), \\
 \lambda & = & \exp{\left[\frac{\pm m \sin{\theta}}{4\text{Re}(a)(A-iB)}\right]} \left(\sqrt{\frac{J}{A-iB}}\right)^{\rho} \Psi^H_L + \exp{\left[\frac{\pm m \sin{\theta}}{4\text{Re}(a)(A+iB)}\right]} \left(\sqrt{\frac{A-iB}{J}}\right)^{\rho} \Psi^H_R, \nonumber \\
\end{eqnarray}
\noindent with $J^{2\sigma} = \exp{\{ \left[2is - \frac{1}{2}(b-\bar{b})\right] S\}} = \exp{\left[ -i\frac{\text{Im}(a)}{2\text{Re}(a)}\ln{J} \right]}$.\footnote{With the symbols $\sigma$, $\omega$ and $\alpha$ with no association with the same symbols used
elsewhere in this work.} Now we will set the following notations for these complex numbers, for the sake of clarity:
\begin{eqnarray}
 \alpha & \equiv & \exp{\left[\frac{iM}{(a+\bar{a})J}\right]}, \\
 \beta & \equiv & J^{2\sigma}, \\
 \delta & \equiv & \sqrt{\frac{J}{A-iB}}, \\
 \epsilon & \equiv & \left(\sqrt{\frac{J}{A-iB}}\right)^\rho, \\
 \omega & \equiv & \exp{\left[\frac{\pm m \sin{\theta}}{4\text{Re}(a)(A-iB)}\right]},\\
 \zeta & \equiv & \exp{\left[\frac{\pm m \sin{\theta}}{4\text{Re}(a)(A+iB)}\right]}.
\end{eqnarray}

In this fashion, one can denote the left- and right-handed components of the fields as
\begin{eqnarray}
 \Psi^D_L & = & \alpha \beta \delta \Psi^H_L,\\
 \Psi^D_R & = & \alpha \beta \delta^{-1} \Psi^H_R,\\
 \lambda_L & = & \epsilon \omega \Psi^H_L,\\
 \lambda_R & = & \epsilon^{-1} \zeta \Psi^H_R,
\end{eqnarray}

which leads to

\begin{eqnarray}
 \lambda_L = \chi_1 \Psi^D_L,\\
 \lambda_R = \chi_2 \Psi^D_R,\\
 \Psi^D_L = \chi_1^{-1} \lambda_L,\\
 \Psi^D_R = \chi_2^{-1} \lambda_R,
\end{eqnarray}

\noindent with the coefficients defined as $\chi_1 \equiv \epsilon \omega \delta^{-1} \beta^{-1} \alpha^{-1}$ and $\chi_2 \equiv \epsilon^{-1} \zeta \delta \beta^{-1} \alpha^{-1}$ being obviously invertible. These coefficients and their inverses are the tools that map Dirac-RIM spinors into Elko-RIM spinors and vice-versa. After some straightforward calculations, one achieves an explicit form of those complex coefficients as


\begin{eqnarray}
 \chi_1 & = & \left(\sqrt{\frac{J}{A-iB}}\right)^{\rho-1} \exp{ \left\{ \frac{1}{2\text{Re}(a)} \left[ \pm \frac{m \sin{\theta}}{2(A-iB)} -i\left( \text{Im}(a)\ln{J} +\frac{M}{J} \right) \right] \right\}}, \nonumber\\ \\
 \chi_1^{-1} & = & \left(\sqrt{\frac{A-iB}{J}}\right)^{\rho-1} \exp{ \left\{ \frac{1}{2\text{Re}(a)} \left[ \mp \frac{m \sin{\theta}}{2(A-iB)} +i\left( \text{Im}(a)\ln{J} +\frac{M}{J} \right) \right] \right\}},\nonumber\\ \\
 \chi_2 & = & \left(\sqrt{\frac{A-iB}{J}}\right)^{\rho-1} \exp{ \left\{ \frac{1}{2\text{Re}(a)} \left[ \pm \frac{m \sin{\theta}}{2(A+iB)} -i\left( \text{Im}(a)\ln{J} +\frac{M}{J} \right) \right] \right\}}, \nonumber\\ \\
 \chi_2^{-1} & = & \left(\sqrt{\frac{J}{A-iB}}\right)^{\rho-1} \exp{ \left\{ \frac{1}{2\text{Re}(a)} \left[ \mp \frac{m \sin{\theta}}{2(A+iB)} +i\left( \text{Im}(a)\ln{J} +\frac{M}{J} \right) \right] \right\}}.\nonumber\\
\end{eqnarray}

This way, one can obtain
\begin{eqnarray}
 \lambda & = & \frac{1}{2}\left[ \chi_1 (\mathbbm{1}+\gamma^5) + \chi_2(\mathbbm{1}-\gamma^5) \right] \Psi^D,\\
 \Psi^D & = & \frac{1}{2}\left[ \chi_1^{-1} (\mathbbm{1}+\gamma^5) + \chi_2^{-1} (\mathbbm{1}-\gamma^5) \right] \lambda.
\end{eqnarray}

If we define the matrices 
\begin{equation}
M \equiv \frac{1}{2}\left[ \chi_1 (\mathbbm{1}+\gamma^5) + \chi_2(\mathbbm{1}-\gamma^5) \right],
\end{equation}
and 
\begin{equation}
N \equiv \frac{1}{2}\left[ \chi_1^{-1} (\mathbbm{1}+\gamma^5) + \chi_2^{-1} (\mathbbm{1}-\gamma^5) \right],
\end{equation}
it easily verifies that $MN = NM = \mathbbm{1}$, i.e., $N = M^{-1}$. Then, we have just proved the following:

\begin{lemma}\label{DiracMDO}
 Let $\varphi_D \in \Pi^D$ and $\varphi_\lambda \in \Pi^M$. There exists a linear isomorphism $M:\Pi^D \rightarrow \Pi^M$, given by means of a matricial operator $M = \frac{1}{2}\left[ \chi_1 (\mathbb{I}+\gamma^5) + \chi_2(\mathbbm{1}-\gamma^5) \right]$, such that
 
\begin{eqnarray}
 \varphi_\lambda & = & M \varphi_D,\\
 \varphi_D & = & M^{-1} \varphi_\lambda.
\end{eqnarray}
 
\end{lemma}


\textit{Lemma} \ref{DiracMDO} shows a linear bijective (algebraic) map between special classes of Elko and Dirac fields, when both are decomposable in terms of RIM-spinors.

Note that an analogue procedure can be done between all the other combinations of the spinor-spaces. Thus, using $(v,w)_\mathcal{A}$ as a notation for the coordinates of a given spinor in a basis $\mathcal{A}$ of a spinor-space $\Pi^\mathcal{A}$, for $\mathcal{A} \in \{ \mathcal{B,D,M} \}$, one can represent $\Psi^H$, $\Psi^D$ and $\lambda$ as

\begin{alignat}{3}
 \Psi^H & = (1,1)_\mathcal{B} &{} = {}& (\alpha^{-1} \beta^{-1} \delta^{-1}, \alpha^{-1} \beta^{-1} \delta)_\mathcal{D} &{} = {}& (\epsilon^{-1} \omega^{-1}, \epsilon \zeta^{-1})_\mathcal{M},\\
 \Psi^D & = (\alpha \beta \delta,\alpha \beta \delta^{-1})_\mathcal{B} &{} = {}& (1,1)_\mathcal{D} &{} = {}& (\chi_1^{-1},\chi_2^{-1})_\mathcal{M},\\
 \lambda & = (\epsilon \omega, \epsilon^{-1} \zeta)_\mathcal{B} &{} = {}& (\chi_1, \chi_2)_\mathcal{D} &{} = {}& (1,1)_\mathcal{M}.
\end{alignat}

Precisely, the construction of the (invertible) operators $L:\Pi^H \rightarrow \Pi^D$ and $Q:\Pi^H \rightarrow \Pi^M$ leads to matricial representations given by

\begin{eqnarray}
 \label{matrizL} L & = & \frac{1}{2}\left[ (\alpha \beta \delta) (\mathbbm{1}+\gamma^5) + (\alpha \beta \delta^{-1}) (\mathbbm{1}-\gamma^5) \right],\\
 \label{matrizQ} Q & = & \frac{1}{2}\left[ (\epsilon \omega) (\mathbbm{1}+\gamma^5) + (\omega^{-1} \zeta) (\mathbbm{1}-\gamma^5) \right],
\end{eqnarray}

\noindent such that 

\begin{eqnarray}
 \label{loke1} \Psi^D & = & L \Psi^H,\\
 \label{loke2} \Psi^H & = & L^{-1} \Psi^D,\\
 \label{loke3} \lambda & = & Q \Psi^H,\\
 \label{loke4} \Psi^H & = & Q^{-1} \lambda.
\end{eqnarray}
 
Then, we can state the following:

\begin{lemma}\label{isos}
 Suppose the existence of a spinor-plane $\Pi^S$ with basis formed by left- and right-handed components of a given spinor $\psi = \psi_L + \psi_R$. If $\psi$ can be decomposed in terms of at least one of $\Psi^H$, $\Psi^D$ or $\lambda$ components with both coefficients non vanishing (in other words, the decomposition is invertible), then it can be written in terms of any of those spinors, i.e., $\Pi^S \cong \Pi^H \cong \Pi^D \cong \Pi^M$.
\end{lemma}

\begin{proof}
 It is trivial, using the results of \textit{Lemma} \ref{DiracMDO} and Equations (\ref{loke1} - \ref{loke4}). 
\end{proof}

Note that \textit{Lemma} \ref{DiracMDO} is a corollary of \textit{Lemma} \ref{isos}.

Another fact that is worthwhile to mention is that $M$, $Q$ and $L$ as shown in \textit{Lemma} \ref{DiracMDO} and Equations (\ref{matrizL}) and (\ref{matrizQ}) are all diagonal (as, obviously, their inverses). This is because of the nature of the chirality projector operators, and we can define:

\begin{definition}
 We define $\mathfrak{M}$ as being the space of all matricial operators $R$ such that $\psi = R \varphi$, with $\psi, \varphi$ being spinors that may be decomposed in terms of RIM-spinors. The space $\mathfrak{M}$ has the set of projector operators $\left\{ \frac{1}{2}(\mathbbm{1}+\gamma^5), \frac{1}{2}(\mathbbm{1}-\gamma^5) \right\}$ as basis, working with complex coefficients to form elements of $\mathfrak{M}$, i.e.,
 
\begin{equation}
 \forall R \in \mathfrak{M}, \exists c_1, c_2 \in \mathbb{C} : R = c_1 \frac{1}{2}(\mathbbm{1}+\gamma^5) + c_2 \frac{1}{2}(\mathbbm{1}-\gamma^5).
\end{equation}
 
 Explicitly, $R = \text{diag}(c_2, c_2, c_1, c_1)$.
 
\end{definition}

It should be clear that, when $c_1,c_2 \neq 0$, every $R \in \mathfrak{M}$ is invertible, with
\begin{equation}
\text{diag}(c_2^{-1}, c_2^{-1}, c_1^{-1}, c_1^{-1}) = R^{-1} \in \mathfrak{M}.
\end{equation}
Finally, given the aspect of all those spinor-planes, we can understand them as being, in fact, exactly the same space, with the matrices $M,L,Q$ and their inverses being change-of-basis matrix operators between the basis $\mathcal{B}$, $\mathcal{D}$ and $\mathcal{M}$, with this being valid for every matrix $R \in \mathfrak{M}$ with other basis of the spinor-plane. This way, we can understand the space $\mathfrak{M}$ as being the space of all change-of-basis matrix operators in the spinor-plane. Then, the spinor-plane is a two-dimensional space of all spinors that may be decomposed in terms of RIM-spinors (given its left- and right-handed components to form a basis on this space), equipped with a space of change-of-basis matrix operators, which encompasses Elko spinors. For more discussions and results concerning the spinor-plane way to approach the study of spinors, please check \cite{restrictedinomata}.

\section{Final Remarks}

In the present communication we have shown the emergence of the Elko spinors - the first concrete example describing mass dimension one fermions with dual helicity feature expected to be a strong candidate to describe dark mater. Interesting enough, we highlight that in its first formulation, we face a non-local theory due to a term, encoded on the spin sums, which are not manifestly invariant or covariant under Lorentz transformations. Such a feature was first believed to be an intrinsic characteristic of dark matter. Nonetheless, non-locality usually pose problems to the physical interpretation, thus, it must to be circumvented. After a while, a mathematical apparatus was developed giving all the necessary support to redefine (or reconstruct) the Elko's dual structure, now, leading to a local and Lorentz invariant theory. As already mentioned, in Ref. \cite{aaca} it is described a subtle way to evade Weinberg's no-go theorem, by exploring another possibility of the dual structure, this time constructing the dual with the additional requirement of Lorentz invariant spin sums. Additional mathematical support was given in \cite{vicinity}. 

Now, regarding to the physical information coming from the bilinear forms, as one can see, we dealt with the necessity to deform the usual Clifford's algebra in order to ascertain the right observance of the Fierz-Pauli-Kofink identities, concerning Elko spinor. Given mathematical programme undergoes on a deformation of the basis of the usual Clifford's algebra looking towards provide new/deformed bilinear forms, taking into account the \emph{Dirac normalisation} method, leading to spinorial densities which hold the FPK identities. As a matter of fact, it must be clear to the reader that Elko can not be fitted or even understood as an object belonging to the Lounesto's classification, a series of reasons are detailed discussed in \cite{beyondlounesto}. Going further, as a net result, only a subset of bilinear covariants experienced are real, bringing some difficulties to the physical interpretation. And in the lights of \cite{beyondlounesto}, thus, we just have a hint towards its possible interpretation.  

After show all the main aspects concerning the Elko construction and also its subtleties, we shall finalize making an allusive comment to one of several applications of Elko spinors. In particular, what we have developed in the so-called spinor-plane is to accommodate a bijective linear map between special classes, i.e., both being RIM-decomposable and, therefore, non exotic, of Elko and Dirac spinors. This mapping is quite natural, as it uses RIM-spinors as a fundamental element making the mediation between Dirac and Elko fields. Although this mapping brings some constraints imposed in the fields themselves, for instance, they have to be RIM-decomposable, one does not have to work with the bilinear covariants, which is often a hard situation to deal with when we study Elko spinors, since they do not necessarily fit in the usual Lounesto's classification. Therefore, the mapping developed in \cite{restrictedinomata} transcends the problem of Lounesto's classification of Elko spinors.
\\\\
R. J. Bueno Rogerio was supported by Conselho Nacional de Desenvolvimento Científico e Tecnológico Grant Number 155675/2018-4, C. H. Coronado Villalobos was supported by Conselho Nacional de Desenvolvimento Científico e Tecnológico Grant Number PCI 300236/2019-0. ARA was partially supported by CAPES.

\section*{Author contribution statement}

\noindent The authors declare that each of the four authors equally contributed to both the
scientific contents and writing of this manuscript.

\appendix

\section{Beyond The Standard Model: Elko from the VSR perspective}\label{apendiceA}

As extensively shown in \cite{horvath} and \cite{ahluwaliahorvath}, Elko spinors can be understood as objects carrying a linear representation of SIM(2) or HOM(2) Lorentz subgroups. In this vein, it is important to explore a possible Clifford algebra basis deformation in this case, if necessary obviously. 

In order to illustrate the situation, we shall make use of the Elko spinors found in \cite{ahluwaliahorvath}. Notice that, in general, the helicity operator does not commute with the VSR boost. Therefore, one cannot freely choose the rest spinors as a basis for such an operator. This is the kernel of the subtle difference between the Elko spinors studied along the main text and the spinors here investigated. 

For an arbitrary momentum, it is possible to see that in the VSR scope we have \cite{ahluwaliahorvath}
\begin{eqnarray}\label{vixi}
\chi^S_{\{-,+\}}(p^{\mu}) = \sqrt{m}\left( \begin{array}{c}
i\frac{p_x-ip_y}{\sqrt{m(p_0-p_z)}}e^{i\phi/2} \\ 
i\sqrt{\frac{p_0-p_z}{m}}e^{i\phi/2} \\ 
\sqrt{\frac{p_0-p_z}{m}}e^{-i\phi/2} \\ 
-\frac{p_x+ip_y}{\sqrt{m(p_0-p_z)}}e^{-i\phi/2}
\end{array} \right),
\end{eqnarray} as a prototype spinor and 
\begin{eqnarray}
\stackrel{\neg}{\chi}^S_{\{-,+\}}(p^{\mu}) = \sqrt{m}\left(\begin{array}{cccc}
0 & -i\frac{e^{-i\phi/2}}{\sqrt{\frac{p_0-p_z}{m}}} & \frac{e^{i\phi/2}}{\sqrt{\frac{p_0-p_z}{m}}} & 0
\end{array} \right),
\end{eqnarray} as its dual, whose construction obeys the same previous prescription, i. e., $\chi^{S/A}_{h}(p^{\mu})=[\Xi(p^{\mu})_{VSR}\chi^{S/A}_{h}(p^{\mu})]^{\dag}\gamma_0$. All the aspects investigated in the main text have a parallel here. The only differences are the general boost of VSR, which reads 
\begin{eqnarray}\label{lalala}
\mathcal{V}_{VSR} = \left(\begin{array}{cccc}
\sqrt{\frac{m}{p_0-p_z}} & \frac{p_1-ip_2}{\sqrt{m(p_0-p_3)}} & 0 & 0 \\ 
0 & \sqrt{\frac{p_0-p_z}{m}} & 0 & 0 \\ 
0 & 0 & \sqrt{\frac{p_0-p_z}{m}} & 0 \\ 
0 & 0 & -\frac{p_1+ip_2}{\sqrt{m(p_0-p_3)}} & \sqrt{\frac{m}{p_0-p_z}}
\end{array}  \right),
\end{eqnarray} and the $\Xi(p^{\mu})_{VSR}$ operator. In order to find this last operator, it is possible to see that starting from Eq. (\ref{XI}), but replacing the usual spinors by VSR Elko spinors (just as (\ref{vixi})) in the composition of $\Xi(p^{\mu})_{VSR}$, one arrives at  
\begin{eqnarray}
\Xi(p^{\mu})_{VSR}= \left( \begin {array}{cccc} {\frac {i \left( p_{{x}}-ip_{{y}}
 \right)e^{i\phi}}{m}}&{\frac {-i \left( p_{{x}}-ip_{{y
}} \right) ^{2}{{\rm e}^{i\phi}}}{m \left( p_{{0}}-p_{{z}} \right) }}-
{\frac {im{{\rm e}^{-i\phi}}}{p_{{0}}-p_{{z}}}}&0&0
\\ \noalign{\medskip}{\frac {i \left( p_{{0}}-p_{{z}} \right) {{\rm e}
^{i\phi}}}{m}}&{-\frac {i \left( p_{{x}}-ip_{{y}}
 \right)e^{i\phi}}{m}}&0&0\\ \noalign{\medskip}0&0&{-\frac {i \left( p_{{x}}+ip_{{y}}
 \right)e^{-i\phi}}{m}}&{\frac {-i \left( p_{
{0}}-p_{{z}} \right) {{\rm e}^{-i\phi}}}{m}}\\ \noalign{\medskip}0&0&{
\frac {i \left( p_{{x}}+ip_{{y}} \right) ^{2}{{\rm e}^{-i\phi}}}{m
 \left( p_{{0}}-p_{{z}} \right) }}+{\frac {i{{\rm e}^{i\phi}}m}{p_{{0}
}-p_{{z}}}}&{\frac {i \left( p_{{x}}+ip_{{y}}
 \right)e^{-i\phi}}{m}}\end{array}\right).\nonumber\\
\end{eqnarray} This operator acts in VSR spinors just as the $\Xi(p^{\mu})$ operator does in usual Elko spinors: it changes the spinor type, obeys $\Xi^{2}(p^{\mu})_{VSR}=\mathbbm{1}$ and have determinant equal to one, ensuring the existence of the inverse. 

With such ingredients, it is possible to see that all the previous steps are repeated here, namely: the FPK identities are not respected, except if the Clifford basis undergoes the same deformation (with $\Xi(p^{\mu})_{VSR}$ replacing $\Xi(p^{\mu})$); $\sigma_{VSR}$ and $K^{\mu}_{VSR}$ are non null real quantities, $\omega_{VSR}$ is zero, and the remain bilinears are imaginary. Besides, as the relation (\ref{eqdiracElko}) still holds, the covariance structure is also the same, provided (\ref{lalala}) instead of the usual Lorentz transformation.

\section{Some properties of Fierz-Pauli-Kofink Identities}\label{apendiceB}
As it can be seen in \cite{bilineares}, it is possible to build the basis vectors for the mass-dimension-one spinor's case using the usual Clifford algebra. For any element $\Gamma$ belonging to such algebra, the FPK relation reads
\begin{eqnarray}
(\stackrel{\neg}{\lambda_{h}} \Gamma \gamma_{\mu}\lambda_{h})=(\stackrel{\neg}{\lambda_{h}} \Gamma \lambda_{h})\lambda_{h}-(\stackrel{\neg}{\lambda_{h}} \Gamma\gamma_{5}\lambda_{h})\gamma_{5}\lambda_{h},
\end{eqnarray}
where $\Gamma \in \{ \mathbbm{1}, \gamma_{5}, \gamma_{\mu}, \Xi(p^{\mu})\gamma_{5}\gamma_{\mu}\Xi(p^{\mu}) \}$. From the above relation we obtain the following:
\begin{eqnarray}
J^2=A^{2}+B^{2},
\end{eqnarray} 
and we also have
\begin{eqnarray}
(\stackrel{\neg}{\lambda_{h}}\Xi(p^{\mu})\gamma_{5}\gamma_{\mu}\Xi(p^{\mu})\lambda_{h})\gamma^{\mu}\lambda_{h}&=&(\stackrel{\neg}{\lambda_{h}}\Xi(p^{\mu})\gamma_{5}\Xi(p^{\mu})\lambda_{h})\lambda_{h}-(\stackrel{\neg}{\lambda_{h}}\Xi^2(p^{\mu})\lambda_{h})\gamma_{5}\lambda_{h},\nonumber\\
&=&(\stackrel{\neg}{\lambda_{h}}\gamma_{5}\lambda_{h})\lambda_{h}-(\stackrel{\neg}{\lambda_{h}}\lambda_{h})\gamma_{5}\lambda_{h}.\label{L}
\end{eqnarray}
Note that
\begin{eqnarray*}
[\Xi(p^{\mu}),\gamma_{5}]=0, \quad \{\gamma_{\mu},\gamma_{5}\}=0 \;\;\;\mbox{and} \quad \Xi^2(p^{\mu})=\mathbbm{1},
\end{eqnarray*}
with such relations at hands, one is able to write
\begin{eqnarray}
(\stackrel{\neg}{\lambda_{h}}\Xi(p^{\mu})\gamma_{5}\gamma_{\mu}\Xi(p^{\mu})\lambda_{h})\gamma^{\mu}\gamma^{5}\lambda_{h}&=&-(\stackrel{\neg}{\lambda_{h}}\Xi(p^{\mu})\gamma_{5}\gamma_{\mu}\Xi(p^{\mu})\gamma_{5}\lambda_{h})\gamma^{\mu}\lambda_{h},\nonumber\\
&=&(\stackrel{\neg}{\lambda_{h}}\Xi(p^{\mu})\gamma_{\mu}\Xi(p^{\mu})\lambda_{h})\gamma^{\mu}\lambda_{h}.\label{LJ}
\end{eqnarray}

Finally using relations \eqref{L} and \eqref{LJ}, we obtain
\begin{eqnarray}
\label{A5}J_{\mu}\gamma^{\mu}\lambda_{L}&=&(A-iB)\lambda_{R},\\
\label{A6}J_{\mu}\gamma^{\mu}\lambda_{R}&=&(A+iB)\lambda_{L}, \\
\label{A7}K_{\mu}\gamma^{\mu}\lambda_{L}&=&-(A-iB)\lambda_{R},\\
\label{A8}K_{\mu}\gamma^{\mu}\lambda_{R}&=&(A+iB)\lambda_{L}.
\end{eqnarray}

\bibliographystyle{unsrt}
\bibliography{refs}
\end{document}